\documentclass{llncs}
\usepackage[nocompress]{cite}
\usepackage{graphicx}
\usepackage{amsmath}
\usepackage[caption=false,font=footnotesize,labelfont=sf,textfont=sf]{subfig}
\usepackage{amssymb, tikz, enumitem}
\usetikzlibrary{math, graphs, quotes}
\usepackage{xspace, setspace, dashbox, mathrsfs}
\usepackage[hidelinks]{hyperref}

\begin{document}

\title{Transcript Franking for Encrypted Messaging}
\author{Armin Namavari \and Thomas Ristenpart}

\institute{Cornell Tech}

\maketitle

\thispagestyle{plain}

\def\authnote{1}

\newcommand{\fixme}[1]{\ifnum\authnote=1{\textcolor{red}{[FIXME: #1]}}\fi}
\newcommand{\todo}[1]{\ifnum\authnote=1{\textcolor{red}{[TODO: #1]}}\fi}
\definecolor{orange}{HTML}{CC5500}
\definecolor{purple}{HTML}{4F00E0}
\newcommand{\armin}[1]{\textcolor{blue}{\ifnum\authnote=1{\sffamily (armin: #1)}}\fi}

\newcounter{mynote}[section]
\renewcommand{\thenote}{\thesection.\arabic{mynote}}

\newcommand{\tnote}[1]{{\color{cyan}[TomR: {{#1}} ]}}

\newcommand{\calO}{\mathcal{O}}
\newcommand{\sig}{\sigma}
\newcommand{\msg}{m}
\newcommand{\sk}{\mathit{sk}}
\newcommand{\pk}{\mathit{pk}}
\newcommand{\frank}{\mathsf{Frank}}
\newcommand{\verify}{\mathsf{Verify}}
\newcommand{\judge}{\mathsf{Judge}}
\newcommand{\redactedMsgFrank}{\mathsf{RMF}}
\newcommand{\transcriptFrank}{\mathsf{TF}}
\newcommand{\bits}{\{0, 1\}}
\newcommand{\bit}{b}
\newcommand{\msgPortion}{\hat{m}}
\newcommand{\msgPortionIdx}{I}
\newcommand{\natNum}{\mathbb{N}}
\newcommand{\sendReq}{\mathsf{Send}}
\newcommand{\syncReq}{\mathsf{Sync}}
\newcommand{\ack}{\mathsf{ack}}
\newcommand{\tagLabel}{\mathsf{tag}}
\newcommand{\macTag}{\mathsf{tag}}
\newcommand{\ct}{\mathsf{ct}}
\newcommand{\servTimstamp}{\mathsf{tserv}}
\newcommand{\recTimestamp}{\mathsf{trec}}
\newcommand{\hash}{H}
\newcommand{\mac}{\mathsf{Tag}}
\newcommand{\macScheme}{\mathsf{MAC}}
\newcommand{\keyGen}{\mathsf{KGen}}
\newcommand{\enc}{\mathsf{Enc}}
\newcommand{\dec}{\mathsf{Dec}}
\newcommand{\ackSend}{\mathsf{AckSend}}
\newcommand{\ackSync}{\mathsf{AckSync}}
\newcommand{\getsr}{{\:{\leftarrow{\hspace*{-3pt}\raisebox{.75pt}{$\scriptscriptstyle\$$}}}\:}}
\newcommand{\encKey}{k}
\newcommand{\macKey}{k_{\mathrm{mac}}}
\newcommand{\msgTsLeq}{\leq_\mathsf{ts}}
\newcommand{\msgCtrLeq}{\leq_\mathsf{ctr}}
\newcommand{\sendCtr}{\mathit{cs}}
\newcommand{\syncCtr}{\mathsf{csync}}
\newcommand{\recvCtr}{\mathit{cr}}
\newcommand{\ctrVec}{\mathit{ctrs}}
\newcommand{\concat}{\mathsf{Concat}}
\newcommand{\party}{P}
\newcommand{\partyZ}{\party_0}
\newcommand{\partyO}{\party_1}
\newcommand{\partyB}{\party_b}
\newcommand{\partyI}{\party_i}
\newcommand{\partyJ}{\party_j}
\newcommand{\partyK}{\party_k}
\newcommand{\partyS}{\party_s}
\newcommand{\partyR}{\party_r}
\newcommand{\group}{\mathscr{G}}
\newcommand{\partyOmb}{\party_{1-b}}
\newcommand{\md}{\mathsf{md}}
\newcommand{\groupId}{\mathit{gid}}
\newcommand{\convId}{\mathit{cid}}
\newcommand{\sender}{\mathsf{sender}}
\newcommand{\receiver}{\mathsf{receiver}}
\newcommand{\rBind}{\mathsf{RecBind}}
\newcommand{\sBind}{\mathsf{SendBind}}
\newcommand{\ctrTable}{\mathsf{ctrTable}}
\newcommand{\timeNow}{\mathsf{TimeNow}}
\newcommand{\oracles}{\mathcal{O}}
\newcommand{\causGraph}{\mathcal{C}}
\newcommand{\lastMsg}{\mathsf{LastMsg}}
\newcommand{\invalidGraph}{\mathsf{invalid}}
\newcommand{\commScheme}{\mathsf{CS}}
\newcommand{\commSp}{\mathcal{Q}}
\newcommand{\Com}{\mathsf{Com}}
\newcommand{\Extr}{\mathsf{Extr}}
\newcommand{\omb}{{1 - b}}
\newcommand{\VerC}{\mathsf{VerC}}
\newcommand{\Ver}{\mathsf{Ver}}
\newcommand{\sendTag}{t_s}
\newcommand{\recvTag}{t_r}

\renewcommand{\paragraph}[1]{\vspace*{4pt}\noindent\textbf{#1.}}

\newcommand{\consref}[1]{Construction~\ref{#1}}
\newcommand{\thmref}[1]{Theorem~\ref{#1}}
\newcommand{\lemmaref}[1]{Lemma~\ref{#1}}
\newcommand{\secref}[1]{Section~\ref{#1}}
\newcommand{\appref}[1]{Appendix~\ref{#1}}
\newcommand{\figref}[1]{Figure~\ref{#1}}
\newcommand{\defref}[1]{Definition~\ref{#1}}
\newcommand{\protref}[1]{Protocol~\ref{#1}}

\newcommand{\tfScheme}{\mathsf{TF}}
\newcommand{\otf}{\mathsf{OTF}}
\newcommand{\setup}{\mathsf{Setup}}
\newcommand{\srvInit}{\mathsf{SrvInit}}
\newcommand{\clInit}{\mathsf{Init}}
\newcommand{\initConv}{\mathsf{InitConv}}
\newcommand{\genSendRequest}{\mathsf{GenSendReq}}
\newcommand{\procSend}{\mathsf{ProcessSend}}
\newcommand{\procSync}{\mathsf{ProcessSync}}
\newcommand{\reportMsgs}{\mathsf{ReportMsgs}}
\newcommand{\verifyMsg}{\mathsf{VerifyMsg}}
\newcommand{\verifyReport}{\mathsf{Judge}}
\newcommand{\compareMsgs}{\mathsf{CmpMsgs}}
\newcommand{\gapSize}{\mathsf{GapSize}}
\newcommand{\isContig}{\mathsf{IsContig}}
\newcommand{\reportInfo}{\rho}
\newcommand{\req}{\mathsf{req}}
\newcommand{\preSendReq}{\mathsf{preSendReq}}
\newcommand{\postSendReq}{\mathsf{postSendReq}}
\newcommand{\postSyncReq}{\mathsf{postSyncReq}}
\newcommand{\st}{\mathit{st}}
\newcommand{\servSt}{\st_S}
\newcommand{\user}{u}
\newcommand{\msgVec}{\mathsf{msgs}}
\newcommand{\reportTag}{\mathsf{tag}}
\newcommand{\advA}{\mathcal{A}}
\newcommand{\advB}{\mathcal{B}}
\newcommand{\advC}{\mathcal{C}}
\newcommand{\advD}{\mathcal{D}}
\newcommand{\advantage}{\mathrm{Adv}}
\newcommand{\advantageEUFCMA}{\advantage^{\textnormal{euf-cma}}}
\newcommand{\advantageVBind}{\advantage^{\textnormal{v-bind}}}
\newcommand{\advantageRBind}{\advantage^{\textnormal{r-bind}}}
\newcommand{\advantageSR}{\advantage^{\textnormal{sr-bind}}}
\newcommand{\advantageSBind}{\advantage^{\textnormal{s-bind}}}
\newcommand{\advantageConf}{\advantage^{\textnormal{conf}}}
\newcommand{\advantageDen}{\advantage^{\textnormal{den}}}
\newcommand{\advantageTrInt}{\advantage^{\textnormal{tr-int}}}
\newcommand{\advantageOTrInt}{\advantage^{\textnormal{o-tr-int}}}
\newcommand{\advantageRep}{\advantage^{\textnormal{tr-rep}}}
\newcommand{\advantageORep}{\advantage^{\textnormal{o-tr-rep}}}
\newcommand{\advantageOFrame}{\advantage^{\textnormal{o-fr}}}
\newcommand{\advantageCorr}{\advantage^{\textnormal{corr}}}
\newcommand{\game}{\mathbf{G}}
\newcommand{\gameSBind}{\game^{\textnormal{rcp-s}}}
\newcommand{\gameRBind}{\game^{\textnormal{rcp-r}}}
\newcommand{\gameSR}{\game^{\textnormal{rcp-sr}}}
\newcommand{\gameTrInt}{\game^{\textnormal{tr-int}}}
\newcommand{\gameOTrInt}{\game^{\textnormal{o-tr-int}}}
\newcommand{\gameOFrame}{\game^{\textnormal{o-fr}}}
\newcommand{\gameRep}{\game^{\textnormal{tr-rep}}}
\newcommand{\gameORep}{\game^{\textnormal{o-tr-rep}}}
\newcommand{\gameCorr}{\game^{\textnormal{corr}}}
\newcommand{\gameConf}{\game^{\textnormal{conf}}}
\newcommand{\gameChanCaus}{\game^{\textnormal{chan-cp}}}
\newcommand{\SendOracle}{\mathbf{Send}}
\newcommand{\RecvOracle}{\mathbf{Recv}}
\newcommand{\TagSendOracle}{\mathbf{TagSend}}
\newcommand{\TagRecvOracle}{\mathbf{TagRecv}}
\newcommand{\SendTagOracle}{\mathbf{SendTag}}
\newcommand{\RecvTagOracle}{\mathbf{RecvTag}}
\newcommand{\SyncOracle}{\mathsf{Sync}}
\newcommand{\ReportOracle}{\mathbf{Rep}}
\newcommand{\ReportReplayOracle}{\mathbf{RepReplay}}
\newcommand{\clen}{\mathsf{clen}}
\newcommand{\queriedCt}{\mathcal{Y}}
\newcommand{\ChalSendOracle}{\mathbf{ChalSend}}
\newcommand{\Init}{\mathsf{Init}}
\newcommand{\userSet}{U}
\newcommand{\requestList}{\mathsf{reqList}}
\newcommand{\ctr}{\mathsf{ctr}}
\newcommand{\ctxt}{\mathit{ctxt}}
\newcommand{\keySp}{\mathcal{K}}
\newcommand{\tagSp}{\mathcal{T}}
\newcommand{\msgSp}{\mathcal{M}}
\newcommand{\ctSp}{\mathcal{C}}
\newcommand{\stSp}{\mathcal{S}}
\newcommand{\setelems}[1]{\{#1\}}
\newcommand{\channel}{\mathsf{Ch}}
\newcommand{\channelSt}{\st_\channel}
\newcommand{\init}{\mathsf{Init}}
\newcommand{\srvInitConv}{\mathsf{SrvInitConv}}
\newcommand{\tagSend}{\mathsf{TagSend}}
\newcommand{\tagRecv}{\mathsf{TagRecv}}
\newcommand{\tagStart}{\mathsf{TagStart}}
\newcommand{\srvTag}{\mathsf{SrvTag}}
\newcommand{\reportReplay}{\mathsf{JudgeReplay}}
\newcommand{\judgeReplay}{\mathsf{JudgeReplay}}
\newcommand{\Snd}{\mathsf{Snd}}
\newcommand{\Rcv}{\mathsf{Rcv}}
\newcommand{\Rprt}{\mathsf{Rprt}}
\newcommand{\GenReport}{\mathsf{GenReport}}
\newcommand{\indexVec}{\mathit{idxs}}
\newcommand{\VerifyReport}{\mathsf{VerifyReport}}
\newcommand{\serverKey}{k_\mathsf{Srv}}
\newcommand{\channelKey}{k_\mathsf{Ch}}
\newcommand{\sentMessages}{\mathcal{R}}
\newcommand{\R}{\mathcal{R}}
\newcommand{\tagggedMessages}{\mathcal{R}_t}
\newcommand{\recvMessages}{\mathcal{R}_r}
\newcommand{\Rf}{\mathcal{R}_f}
\newcommand{\localG}{\mathsf{localG}}
\newcommand{\mfcFBCaus}{\mathsf{MFCh}_\mathrm{cFB}}
\newcommand{\mfc}{\mathsf{MFCh}}
\newcommand{\mfcAck}{\mathsf{MFCh}_\mathrm{ack}}
\newcommand{\barP}{{\bar{P}}}
\newcommand{\winFlag}{\mathsf{win}}
\newcommand{\sendOp}{\texttt{S}}
\newcommand{\recvOp}{\texttt{R}}
\newcommand{\vertexSpace}{\mathcal{V}}
\newcommand{\edgeSpace}{\mathcal{E}}
\newcommand{\Z}{\mathbb{Z}}

\newcommand{\cif}{\mathbf{if\;}}
\newcommand{\cthen}{\mathbf{\;then\;}}
\newcommand{\celse}{\mathbf{else\;}}
\newcommand{\creturn}{\mathbf{return\;}}
\newcommand{\ctrue}{\mathsf{true}}
\newcommand{\cfalse}{\mathsf{false}}
\newcommand{\cbad}{\mathsf{bad}}
\newcommand{\cflag}{\mathsf{flag}}
\newcommand{\sets}{\mathsf{sets\;}}

\newcommand{\ind}{\hspace*{1.5em}}
\newcommand{\indsm}{\hspace*{.75em}}
\newcommand{\indeqn}{\;\;\;\;\;\;\;}

\newcommand{\query}[1]{\procfont{query} {#1}:}
\newcommand{\queryl}[1]{\underline{\procfont{query} {#1}:}}
\newcommand{\procedure}[1]{\underline{\procfont{procedure} {#1}:}}
\newcommand{\procedurev}[1]{\underline{{#1}:}\smallskip}
\newcommand{\subroutine}[1]{\underline{\procfont{subroutine} {#1}:}}
\newcommand{\subroutinev}[1]{\underline{\procfont{subroutine} {#1}:}\smallskip}
\newcommand{\subroutinenl}[1]{{\procfont{subroutine} {#1}:}}
\newcommand{\subroutinenlv}[1]{{\procfont{subroutine} {#1}:}\smallskip}
\newcommand{\adversary}[1]{\underline{\procfont{adversary} {#1}:}}
\newcommand{\adversaryv}[1]{\underline{\procfont{adversary} {#1}:}\smallskip}
\newcommand{\experiment}[1]{\underline{{#1}}}
\newcommand{\experimentv}[1]{\underline{{#1}}\smallskip}
\newcommand{\algorithm}[1]{\underline{\procfont{algorithm} {#1}:}}
\newcommand{\algorithmv}[1]{\underline{\procfont{algorithm} {#1}:}\smallskip}

\newcommand{\gamesfontsize}{\scriptsize}
\newcommand{\stretchval}{1.2}

\newcommand{\mpage}[2]{\begin{minipage}{#1\textwidth} #2 \end{minipage}}
\newcommand{\fpage}[2]{\framebox{\begin{minipage}{#1\textwidth}\setstretch{\stretchval}\gamesfontsize #2 \end{minipage}}}
\newcommand{\opage}[2]{\begin{minipage}{#1\textwidth}\setstretch{\stretchval}\gamesfontsize #2 \end{minipage}}

\newcommand{\fhpages}[3]{
    \framebox{
        \begin{tabular}{cc}
            \begin{minipage}[t]{#1\textwidth}\setstretch{\codestretch}\gamesfontsize #2 \end{minipage}
             &
            \begin{minipage}[t]{#1\textwidth}\setstretch{\codestretch}\gamesfontsize #3 \end{minipage}
        \end{tabular}
    }
}
\newcommand{\hfpages}[3]{\hfpagess{#1}{#1}{#2}{#3}}
\newcommand{\hfpagess}[4]{
    \begin{tabular}{c@{\hspace*{.5em}}c}
        \framebox{\begin{minipage}[t]{#1\textwidth}\setstretch{\stretchval}\gamesfontsize #3 \end{minipage}}
         &
        \framebox{\begin{minipage}[t]{#2\textwidth}\setstretch{\stretchval}\gamesfontsize #4 \end{minipage}}
    \end{tabular}
}
\newcommand{\hfpagesss}[6]{
    \begin{tabular}{c@{\hspace*{.5em}}c@{\hspace*{.5em}}c}
        \framebox{\begin{minipage}[t]{#1\textwidth}\setstretch{\stretchval}\gamesfontsize #4 \end{minipage}}
         &
        \framebox{\begin{minipage}[t]{#2\textwidth}\setstretch{\stretchval}\gamesfontsize #5 \end{minipage}}
         &
        \framebox{\begin{minipage}[t]{#3\textwidth}\setstretch{\stretchval}\gamesfontsize #6 \end{minipage}}
    \end{tabular}
}
\newcommand{\hfpagessss}[8]{
    \begin{tabular}{c@{\hspace*{.5em}}c@{\hspace*{.5em}}c@{\hspace*{.5em}}c}
        \framebox{\begin{minipage}[t]{#1\textwidth}\setstretch{\stretchval}\gamesfontsize #5 \end{minipage}}
         &
        \framebox{\begin{minipage}[t]{#2\textwidth}\setstretch{\stretchval}\gamesfontsize #6 \end{minipage}}
         &
        \framebox{\begin{minipage}[t]{#3\textwidth}\setstretch{\stretchval}\gamesfontsize #7 \end{minipage}}
         &
        \framebox{\begin{minipage}[t]{#4\textwidth}\setstretch{\stretchval}\gamesfontsize #8 \end{minipage}}
    \end{tabular}
}

\newcommand{\hfpagesssss}[6]{
    \begin{tabular}{c@{\hspace*{.5em}}c@{\hspace*{.5em}}c@{\hspace*{.5em}}c@{\hspace*{.5em}}c}
        \framebox{\begin{minipage}[t]{#1\textwidth}\setstretch{\stretchval}\gamesfontsize #2 \end{minipage}}
         &
        \framebox{\begin{minipage}[t]{#1\textwidth}\setstretch{\stretchval}\gamesfontsize #3 \end{minipage}}
         &
        \framebox{\begin{minipage}[t]{#1\textwidth}\setstretch{\stretchval}\gamesfontsize #4 \end{minipage}}
         &
        \framebox{\begin{minipage}[t]{#1\textwidth}\setstretch{\stretchval}\gamesfontsize #5 \end{minipage}}
         &
        \framebox{\begin{minipage}[t]{#1\textwidth}\setstretch{\stretchval}\gamesfontsize #6 \end{minipage}}
    \end{tabular}
}

\newcommand{\fhpagesss}[6]{
    \framebox{
        \begin{tabular}{c@{\hspace*{.5em}}c@{\hspace*{.5em}}c}
            \begin{minipage}[t]{#1\textwidth}\setstretch{\stretchval}\gamesfontsize #4 \end{minipage}
             &
            \begin{minipage}[t]{#2\textwidth}\setstretch{\stretchval}\gamesfontsize #5 \end{minipage}
             &
            \begin{minipage}[t]{#3\textwidth}\setstretch{\stretchval}\gamesfontsize #6 \end{minipage}
        \end{tabular}
    }
}

\def\codestretch{\stretchval}

\newcommand{\hpagesl}[3]{
    \begin{tabular}{c|c}
        \begin{minipage}{#1\textwidth}\setstretch{\codestretch} #2 \end{minipage}
         &
        \begin{minipage}{#1\textwidth} #3 \end{minipage}
    \end{tabular}
}

\newcommand{\hpagessl}[4]{
    \begin{tabular}{c|@{\hspace*{.5em}}c}
        \begin{minipage}[t]{#1\textwidth}\setstretch{\codestretch} #3 \end{minipage}
         &
        \begin{minipage}[t]{#2\textwidth}\setstretch{\codestretch} #4 \end{minipage}
    \end{tabular}
}

\newcommand{\hpages}[3]{
    \begin{tabular}{cc}
        \begin{minipage}[t]{#1\textwidth}\setstretch{\codestretch} #2 \end{minipage}
         &
        \begin{minipage}[t]{#1\textwidth}\setstretch{\codestretch} #3 \end{minipage}
    \end{tabular}
}

\newcommand{\hpagess}[4]{
    \begin{tabular}{c@{\hspace*{1.5em}}c}
        \begin{minipage}[t]{#1\textwidth}\setstretch{\codestretch} #3 \end{minipage}
         &
        \begin{minipage}[t]{#2\textwidth}\setstretch{\codestretch} #4 \end{minipage}
    \end{tabular}
}

\newcommand{\hpagesss}[6]{
    \begin{tabular}{ccc}
        \begin{minipage}[t]{#1\textwidth}\setstretch{\codestretch}\gamesfontsize #4 \end{minipage} &
        \begin{minipage}[t]{#2\textwidth}\setstretch{\codestretch}\gamesfontsize #5 \end{minipage} &
        \begin{minipage}[t]{#3\textwidth}\setstretch{\codestretch}\gamesfontsize #6 \end{minipage}
    \end{tabular}}
\newcommand{\hpagesssl}[6]{
    \begin{tabular}{c|c|c}
        \begin{minipage}[t]{#1\textwidth}\setstretch{\codestretch} #4 \end{minipage} &
        \begin{minipage}[t]{#2\textwidth}\setstretch{\codestretch} #5 \end{minipage} &
        \begin{minipage}[t]{#3\textwidth}\setstretch{\codestretch} #6 \end{minipage}
    \end{tabular}}
\newcommand{\hpagessss}[8]{
    \begin{tabular}{cccc}
        \begin{minipage}[t]{#1\textwidth}\setstretch{\codestretch} #5 \end{minipage} &
        \begin{minipage}[t]{#2\textwidth}\setstretch{\codestretch} #6 \end{minipage} &
        \begin{minipage}[t]{#3\textwidth}\setstretch{\codestretch} #7 \end{minipage}
        \begin{minipage}[t]{#4\textwidth}\setstretch{\codestretch} #8 \end{minipage}
    \end{tabular}}

\begin{abstract}
    Message franking is an indispensable abuse mitigation tool for end-to-end
    encrypted (E2EE) messaging platforms. With it, users who receive harmful
    content can securely report that content to platform moderators. However,
    while real-world deployments of reporting require the disclosure of multiple
    messages,  existing treatments of message franking only consider the report
    of a single message. As a result, there is a gap between the security goals
    achieved by constructions and those needed in practice.
    
    Our work introduces \emph{transcript franking}, a new type of protocol that
    allows reporting subsets of conversations such that moderators can
    cryptographically verify message causality and contents. We define syntax,
    semantics, and security for transcript franking in two-party and group
    messaging. We then present efficient constructions for transcript franking
    and prove their security. Looking toward deployment considerations, we
    provide detailed discussion of how real-world messaging systems can
    incorporate our protocols.
\end{abstract}

\section{Introduction}
\label{sec:intro}

End-to-end encrypted (E2EE) messaging is used by billions of people through
platforms like Whatsapp, Signal, and iMessage \cite{whatsapp-users}. As a
result, users enjoy strong security and privacy protections even in the face of
messaging platform compromise by malicious insiders, remote attackers, or
government overreach. Abuse, hate, and harassment, however, are not prevented or
mitigated by encryption, and encrypted messaging platforms are used to spread
misinformation, incitements of violence, and illegal content
\cite{whatsapp-misinfo}. As a result, a rapidly growing body of work has sought
to provide trust and safety features for encrypted messaging, without
diminishing its privacy benefits \cite{sokContentModeration}.

One important line of work targets secure reporting of abusive  
messages (see~\cite{sokContentModeration}). When users receive harmful content,
they can report it to the platform, which can in turn take appropriate action
against the user that sent the problematic content. Such user-driven reporting
features are widespread on plaintext platforms and play an instrumental role in
content moderation across a wide range of abuse types~\cite{contentObliv}. In
encrypted messaging, however, the platform cannot trivially verify that a report
corresponds to a transmitted message. 

Facebook's message franking feature~\cite{fbSecretConv} was the first to target
cryptographically verifiable abuse reports. Message franking targets not
compromising the confidentiality of unreported messages, and preventing attacks
that undermine the trustworthiness of reporting: users should not be able to
report messages that were not sent, nor be able to send messages that cannot be
reported. These security properties were first formalized in~\cite{smfComm} as
receiver binding and sender binding, respectively. While Facebook's first design
had a sender binding vulnerability~\cite{invisSalamanders}, we now have message
franking protocols with strong assurance in their
security~\cite{invisSalamanders,smfComm,msgFrChan}. Subsequent work extended to
provide asymmetric message franking schemes (AMFs)~\cite{amf,hecate} for
two-party sender-anonymous messaging, group AMFs~\cite{groupAMF}, franking for
two-party channels~\cite{msgFrChan}, and message franking that allows only
revealing parts of messages~\cite{leontiadis2023private}.

All those treatments of message franking only support reporting individual
messages. In practice, however, moderators typically need visibility into more
of a conversation to make judgements, and indeed existing abuse reporting
workflows do report surrounding messages when one is
reported~\cite{smfComm,whatsappReport}. Recent work~\cite{reportingHCI} reports
that users would find it useful to have more agency in specifying what portions
of their private conversations are disclosed to moderators, which is not
something current approaches offer. 

Despite this, to date, there has been no attempt to show how to provide
cryptographically verifiable reporting of multi-message conversations.
Near-at-hand approaches, including those used in practice, do not provide a
satisfying level of security. Consider reporting with message franking: each
individual message can be verified along with a platform timestamp of when it
was sent. But a malicious client can simply undetectably omit messages from a
report. For example consider the conversation between Alice and Bob shown in
\figref{fig:message_order}. If Alice reports the conversation with omission of
$m_1$ it blocks moderators from interpreting Bob's message $m_3$ as replying to
$m_1$ rather than mocking Alice's loss. A more subtle issue is that even
high-precision timestamps do not establish strict causal
ordering~\cite{lamportClock}. Let $m < m'$ represent that~$m$ was received
before~$m'$ was sent. Returning to the figure, it could be that $m_2 < m_3$, or
it could be that $m_3$ was sent before Bob received $m_2$. The result in the
latter case would be Bob's view of the conversation being different from
Alice's.

\def\w{4} \def\h{1} \def\rightShift{9.7} \def\tw{3.5} \def\roundVal{7}
\def\sep{0.1}

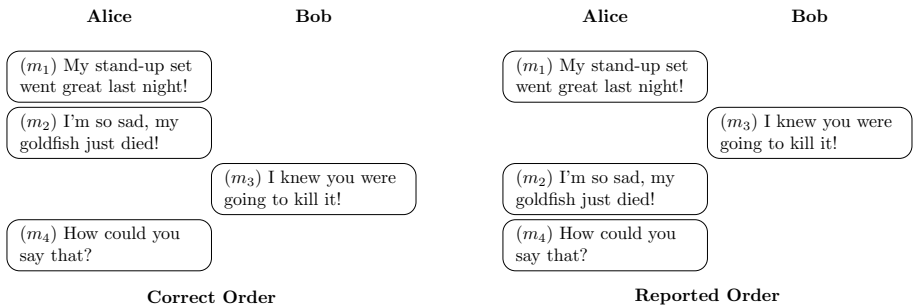
\begin{figure}[t]
    \centering
    \resizebox{\textwidth}{!}{
    \begin{tikzpicture}
        \node at (0.5 * \w cm, 5 * \h) {\bf Alice}; \node at (1.5 * \w cm, 5 *
        \h) {\bf Bob}; \node at (\w cm, -0.5 * \h) {\bf Correct Order};

        \draw [rounded corners = \roundVal] (0 cm, 0 cm) rectangle node[text
        width=\tw cm, align=left]{$(m_4)$ How could you say that?} (\w cm,\h cm);
        
        \draw [rounded corners = \roundVal, yshift = \sep cm] (\w cm,\h cm)
     rectangle node[text width=\tw cm, align=left]{ $(m_3)$ I knew you
     were going to kill it!} (2*\w cm,2*\h cm);

        \draw [rounded corners = \roundVal, yshift = 2*\sep cm] (0 cm,2*\h cm)
        rectangle node[text width=\tw cm, align=left]{ $(m_2)$ I'm so sad,
        my goldfish just died!} (\w cm,3*\h cm);

        \draw [rounded corners = \roundVal, yshift = 3*\sep cm] (0 cm,3*\h cm)
        rectangle node[text width=\tw cm, align=left]{ $(m_1)$ My stand-up
        set went great last night!} (\w cm,4*\h cm);

        \node [xshift = \rightShift cm] at (0.5 * \w cm, 5 * \h) {\bf Alice};
        \node [xshift = \rightShift cm] at (1.5 * \w cm, 5 * \h) {\bf Bob};
        \node [xshift = \rightShift cm] at (\w cm, -0.5 * \h) {\bf Reported Order};

        \draw [rounded corners = \roundVal, xshift = \rightShift cm] (0 cm, 0
        cm) rectangle node[text width=\tw cm, align=left]{ $(m_4)$ How could you
        say that?}
        (\w cm,\h cm);
        
        \draw [rounded corners = \roundVal, xshift = \rightShift cm, yshift =
     2*\sep cm] (\w cm,2*\h cm) rectangle node[text width=\tw cm,
     align=left]{ $(m_3)$ I knew you were going to kill it!} (2*\w cm,3*\h
     cm);

        \draw [rounded corners = \roundVal, xshift = \rightShift cm, yshift =
        \sep cm] (0 cm,\h cm) rectangle node[text width=\tw cm,
        align=left]{ $(m_2)$ I'm so sad, my goldfish just died!} (\w
        cm,2*\h cm);

        \draw [rounded corners = \roundVal, xshift = \rightShift cm, yshift =
        3*\sep cm] (0 cm,3*\h cm) rectangle node[text width=\tw cm,
        align=left]{ $(m_1)$ My stand-up set went great last night!} (\w
        cm,4*\h cm);

    \end{tikzpicture}
    }

    \label{fig:message_order}
    \caption{An example conversation in which message ordering impacts
    interpretation. A report of this conversation should confirm for moderators
    the causal ordering. }
\end{figure}

The problem of conversation ordering and moderation has been known to
practitioners since at least 2014~\cite{signalBlogIceCream}, but only recently
has there been a first effort to address it by Chen and
Fischlin~\cite{causalityPres}. They propose a message franking protocol,
$\mfcFBCaus$, in which clients report observed message ordering via franking
metadata alongside content. While their approach provides a novel augmentation
of message franking with causality, they stop short of providing a fleshed out
solution for multi-message franking. Their security modeling considers reporting
individual received messages. Meanwhile, allowing reporters to disclose sent
messages is crucial to multi-message reporting with full context. One could
consider a natural extension of $\mfcFBCaus$ in which the report function is
invoked for each message and both parties are involved in disclosing each
other's messages. This, however, is susceptible to denial of service as the
abusive party can simply refuse to cooperate and go offline, withholding crucial
context.

Moreover, $\mfcFBCaus$ is susceptible to integrity attacks. Since causality data
is client-generated, malicious clients can provide incorrect information about
message ordering. Consider the example conversation in
\figref{fig:message_order}. We demonstrate an attack in which Bob makes Alice
observe the view on the left while making it so that Alice can only report the
view on the right. Bob simply sends $m_3$ after having received $m_2$ but
attaches causality metadata to $m_3$ indicating he has seen only $m_1$. Hence,
Bob can force Alice to observe an upsetting message ordering while making seem
as if it were due to the network reordering messages.

\paragraph{Our contributions} We suggest a new approach that we call
\emph{transcript franking}. This cryptographic protocol goal allows users to
report some or all of two-party or group conversations with stronger security
guarantees about message ordering and omissions. We define the syntax and
semantics of a transcript franking scheme and provide formal security
definitions. We go on to detail transcript franking schemes for both the
two-party and group messaging settings; our schemes are practical to deploy and
avoid issues like those above, clarifying in a single report all relevant
information about the messages reported, including causal ordering. To do so,
our schemes deviate from Chen-Fischlin's approach of using client-provided
causality metadata, instead taking advantage of the fact that the platform can
establish causal ordering over ciphertexts and check that reports are consistent
with it. We prove that our new schemes meet our new security goals under
standard assumptions. 

We treat both two-party (direct) and group messaging settings; we explain
further each, starting with the former.

\paragraph{Two-party transcript franking} We start with the two-party case,
where only two clients are involved in a conversation (sometimes called direct
messaging). Our goal is to enable either participant in the encrypted
conversation to report all or part of the conversation to the platform. Our
starting point is symmetric message franking (SMF)~\cite{smfComm,msgFrChan}, in
which the clients encrypt messages using committing authenticated encryption
with associated data (AEAD) and the platform MACs (portions of) the resulting
ciphertexts to produce a reporting tag. To report a message, the ciphertext,
secret encryption key, and the reporting tag are all sent to the provider. 

In this setting where we rely on fast symmetric primitives, the veracity of a
ciphertext can only be checked by those who have access to the secret
keys---before reporting, just the two clients. Prior work on SMFs took this to
mean that one cannot support reporting a sent message, since the reporter could
be unreliable. But if we want to allow reporting more of a conversation, we must
support this. Intuitively, our approach will be to leverage acknowledgements of
received ciphertexts to register their validity for their use in reporting.

To do so, we first enrich the architectural model to surface how the provider
manages sending and receiving events separately. This more accurately captures
the asynchronous nature of messaging.  Formally, the platform is represented by
a pair of algorithms $\tagSend$ and $\tagRecv$.  We allow the platform to
maintain per-conversation state; we also give a way to outsource this state
securely to clients. A client sends a message by running an algorithm~$\Snd$ and
submitting the resulting ciphertext to the platform by calling~$\tagSend$. To
receive a message, a recipient client runs an algorithm $\Rcv$ on it, and if it
accepts the message, indicates so to the server, which then calls $\tagRecv$.
The latter readies a cryptographic reception acknowledgement for both the sender
and the receiver.  

A set of messages can be reported by submitting their ciphertexts, their
corresponding secret keys, and platform-provided reporting tags. The moderator
can process them via an algorithm~$\verifyReport$ that verifies them, and
produces a causal graph providing the moderator with: (1) a partial ordering
over all messages that implies a total ordering over messages sent in either
direction; and (2) indication of gaps---unreported messages sent between
reported messages.

We require that transcript franking schemes provide confidentiality and
integrity of messages that aren't reported in the face of a malicious provider.
More complex is the security of reporting, in which the platform is trusted, but
clients are not. Here we formalize two security goals for transcript franking
schemes, which can be viewed as strengthening SMF sender and receiver binding.
Coming up with definitions that capture transcripts, rather than individual
messages, proved challenging.  For example, any given conversation can now give
rise to all sorts of possible reports: the entire conversation or any subset of
that conversation, including possibly just a single message. 

Our first definition is \emph{transcript reportability}. Here the security game
tasks an adversary controlling one client with interacting with the provider and
some other honest client. The adversary attempts to generate a message
transcript such that the honest client cannot successfully report some
adversarially-chosen subset of messages. In the case where only a single message
is reported, this coincides with sender binding, but it goes much beyond since
it requires that any subset of the conversation be reportable, including
messages sent by the reporter. 

The second main security notion is \emph{transcript integrity}. Intuitively we
want to not only ensure that no maliciously generated report can frame someone
as having sent something when they have not, but, moreover, ensure that the
reported transcript does not lie about ordering or omissions. 
Perhaps counterintuitively, we increase the adversarial power over prior
treatments of receiver binding by allowing the adversary to control \emph{both}
clients in a conversation to arbitrarily deviate from the protocol. The
adversary has oracles to send ciphertexts to the platform and register having
received them. The game keeps track of a ground truth graph of message
transmission. This matches the view of the platform in terms of sending and
receiving event orderings for each transmitted ciphertext. Then, the adversary's
goal is to come up with either a report that generates a causality graph
inconsistent with the ground truth graph, or two reports that are inconsistent
with each other. We define consistency relative to our causality graph
abstraction.

\paragraph{The quad-counter construction} We now turn to constructions.  We want
to ensure practicality, using fast cryptographic mechanisms and avoiding
unrealistic storage constraints. The main underlying idea is to have the
platform carefully attest to the observed causal order of ciphertexts.  Since we
allow stateful platforms, we could of course just store a log of all sending and
reception events that occurred, with their corresponding ciphertexts. But this
would be prohibitively expensive in terms of storage. Instead, we can use MAC'd
counters of events. Namely, for each party we have a pair of counters, a sending
counter incremented each time that party sends a ciphertext, and a reception
counter incremented each time a ciphertext is registered as received. The
platform generates a cryptographic acknowledgement for each send and receive
event: a tuple including the party identifiers, the type of event (send or
receive), a binding commitment to the ciphertext, and the current counters,
together with a MAC over the tuple using a platform-held secret key.
Cryptographic acknowledgements are shared with both parties. The platform need
only retain the four counters, hence the name of this protocol as the
quad-counter construction (QCC). 

A reporter can choose any subset of messages, and submits the identity of the
sender, the message itself, the cryptographic commitment to the message, and
both the sending and reception cryptographic acknowledgements. The platform can
then verify each tuple, and construct a causality graph that indicates the
precise causality order of the reported messages plus how many send and receive
events were omitted from the report.

We show that, under standard assumptions on the MAC and commitment scheme, we
achieve transcript integrity.

\paragraph{Extensions} The quad-counter construction readily extends to work in
group messaging settings, by having a sending and reception counter for each
party (for a total of~$2n$ counters where $n$ is the number of participants).
The cryptographic acknowledgements are otherwise similar to the two-party case.
We formalize this setting, showing how our definitions and results handle it
securely. 

One potential downside of the quad-counter construction and its generalization,
as compared to in-use (but insecure) approaches like timestamps, is that the
platform must now keep per-conversation state. While this state is tiny, it
would be better to avoid, due to it complicating large-scale deployments where
one would prefer to have platform servers stateless and able to load balance any
message event across any server without having to synchronize state. We describe
an extension to our approach which outsources the state to the clients, at the
cost of slightly extra data being sent to the platform, and relying on honest
users to report replay of cryptographic acknowledgements. See
\secref{sec:discussion} for details.

\section{Preliminaries}
\label{sec:prelim}
\subsection{General Notation and Primitives}
Let $\Z^*$ denote the non-negative integers $\setelems{0, 1, \ldots}$. To
indicate the range of elements $\setelems{0, 1, \ldots, N-1}$, we use the
shorthand $[N]$. The symbol $\lambda \in \Z^*$ denotes the security parameter.
For an adversary $\advA$ in a game $\game$, we use $\Pr[\game(\advA) \Rightarrow
x]$ to denote the probability that the $\game$ outputs $x$ when run with
adversary $\advA$. With a tuple, we refer to its elements via 0-indexed
positions or names. For instance, $c.x$ or $c[0]$ refer to the first element of
the tuple $c = (x, y, z)$. Multiple indexing is also allowed: $c.(x, z)$
extracts the first and third elements of the tuple, while $c[1:2]$ extracts the
last two elements (for ranges, both the start and end indices are inclusive).
We indicate structs as sets of variables, e.g., $s \gets \{x, y, z\}$. When $s$
is in scope, we allow accessing $x$ directly in order to simplify notation.
The notation $d \gets D(x)$ indicates obtaining the output $d$ from a
deterministic algorithm $D$, when fed input $x$. For a randomized algorithm $R$,
we write $r \getsr R(x)$ to denote obtaining the output $r$ from the input $x$.
We utilize the terms \emph{algorithm}, \emph{routine}, and \emph{procedure}
interchangeably.

\paragraph{Bidirectional channels}
To model two-party communication, we make use of a bidirectional channel
abstraction, borrowing syntax from \cite{causalityPres}. A bidirectional channel
$\channel$ consists of three procedures $\channel = (\init, \Snd, \Rcv)$,
defined over a key space $\keySp$, a message space $\msgSp$, a party space
$\bits$, a ciphertext space $\ctSp$, and a state space $\stSp$. The variable $P$
and the labels $\setelems{0, 1}$ are used to indicate the two parties. Let the
notation $\barP$ be shorthand for $1 - P$. We elaborate on these procedures
below.
\begin{itemize}
    \item $\st \gets \init(P, k)$ generates initial channel state $\st \in \stSp$ for
    party $P \in \bits$ with key $k \in \keySp$.
    \item $\st', c \getsr \Snd(P, \st, m)$ produces a ciphertext $c \in \ctSp$
    from party $P \in \bits$ for plaintext $m \in \msgSp$ and client state $\st$,
    yielding updated state $\st'$.
    \item $\st', m, i \gets \Rcv(P, \st, c)$ decrypts a ciphertext $c \in
    \ctSp$, received by the party $P \in \bits$, to plaintext $m \in \msgSp$
    along with a sending index $i \in \Z^*$.
\end{itemize}

Correctness of a channel requires that all honestly sent messages can be
successfully received with the correct sending index. For security, we expect
standard confidentiality and integrity properties. See \cite{causalityPres} for
more details.

\paragraph{Message authentication code} A message authentication code (MAC)
consists of the algorithms $\macScheme = (\keyGen, \mac, \verify)$ defined over a key space
$\keySp$, a message space $\msgSp$, and a tag space $\tagSp$. The key generation
procedure $\keyGen$ outputs a random key $k \in \keySp$. The procedure $\mac(k,
m)$ outputs a tag $t \in \tagSp$ for a message $m \in \msgSp$. To verify a tag
$t$ on a message $m$, one runs the procedure $\verify(k, m, t)$, which outputs
$b \in \bits$. An output of 1 indicates a valid tag on the message while an
output of 0 indicates an invalid tag. For a MAC to be considered secure, it has
to satisfy existential unforgeability under a chosen message attack (EUF-CMA).
This means that, when given oracle access to $\mac(k, \cdot)$ and $\verify(k,
\cdot, \cdot)$, for $k \getsr \keyGen()$, an adversary $\advA$ has a negligible
probability of producing $(m, t)$, where $m$ is not previously queried to
$\mac$, that passes the verification check. We denote this probability as the
advantage $\advantageEUFCMA_{\macScheme}(\advA)$.

\paragraph{Commitment scheme}
A commitment scheme $\commScheme$ consists of a pair of algorithms $(\Com,
\VerC)$ defined over a message space $\msgSp$, a key space $\keySp$, and a
commitment space $\commSp$. The randomized algorithm $\Com(m)$ outputs a pair
$(k, c) \in \keySp \times \commSp$, where $c$ is the commitment and $k$ is a key
that allows opening the commitment to the original message $m \in \msgSp$. In
terms of security, a commitment scheme is often required to be \emph{hiding} and
\emph{binding}. To satisfy the hiding property, the commitment $c$ must reveal
nothing about the message $m$. We formalize this via a real-or-random notion
that requires that no efficient adversary can distinguish a commitment from a
random bit-string of the same length. The binding property requires that an
adversary $\advA$ has a negligible probability in producing a tuple $(m, k, m',
k', c)$, where $m \neq m'$, $\VerC(m, k, c) = 1$, and $\VerC(m', k', c) = 1$.

\paragraph{Message franking}
User-driven reporting for end-to-end encrypted messaging is captured by message
franking \cite{smfComm,fbSecretConv,invisSalamanders}. For now, we focus our
attention on two-party conversations between users. For the sake of notational
simplicity, we elide associated metadata that clients or the server may
associate with the message (e.g., timestamps). In accordance with our usage of
messaging channels, we draw on message franking channels \cite{msgFrChan}. We
opt for the syntax used in Chen and Fischlin's work \cite{causalityPres}. The
procedure $\Rcv(\party, \st, c)$ outputs $\st', m, k_f, i$, where $k_f$ is a key
opening the commitment $c.c_f$, which is stored as part of the ciphertext $c$.
There is also a server tagging procedure $\srvTag(\servSt, P, c_f)$, which
outputs a tag $t$ on a franking commitment $c_f$, which corresponds to a
ciphertext sent by party $P$. A reporting procedure $\Rprt(\servSt, P, m, k_f,
c_f, t)$ verifies for the server that $t$ is a server tag on $c_f$, which opens
to a message $m$.

\subsection{Causality Graphs}
In order to model message transmission in a way that accounts for asynchronous
networks, we use causality graphs. We define our causality graph object here,
which draws heavily on the graph formalism used in \cite{causalityPres}. A
causality graph $G$ modeling two-party messaging is a directed bipartite graph
consisting of two disjoint sets of vertices $V_0$ and $V_1$ and an edge set $E$.
An edge is a pair of vertices $(v_1, v_2)$ where $v_1 \in V_P$ and $v_2 \in
V_\barP$ for some $P \in \bits$. We write $V = V_0 \cup V_1$ to refer to the
full set of vertices within the graph. Each vertex $v \in V$ contains four
pieces of data: an action type $t \in \setelems{S, R}$, a sending counter
$\sendCtr$, a reception counter $\recvCtr$, and a message $m \in \msgSp$.
Indeed, these four pieces of data are sufficient to uniquely identify a vertex
within a single conversation. We can therefore define a vertex space
$\vertexSpace$ as the direct product $\setelems{S, R} \times \Z^{*} \times
\Z^{*} \times \msgSp$, and the edge space $\edgeSpace$ as $\vertexSpace \times
\vertexSpace$. The action type $t$ indicates whether the vertex corresponds to a
sending ($S$) event or a reception ($R$) event. For any message sent from $P$ to
$\barP$, there is an edge from a sending vertex in $P$ to a receiving vertex in
$\barP$. We define the sub-graph relation as follows: for two graphs $G_1 =
(V_1, E_1)$ and $G_2 = (V_2, E_2)$, we write $G_1 \subseteq G_2$ if $V_1
\subseteq V_2$ and $E_1 \subseteq E_2$. We will also allow causality graphs that
do not contain messages, which will become useful for our security definitions.
Let $G$ be a causality graph; we denote the message-excluded version of $G$ as
$\widetilde{G} = (\setelems{v[0:2] \mid v \in G.V}, \setelems{(v_1[0:2],
v_2[0:2]) \mid (v_1, v_2) \in G.E})$.

\begin{figure}[ht]
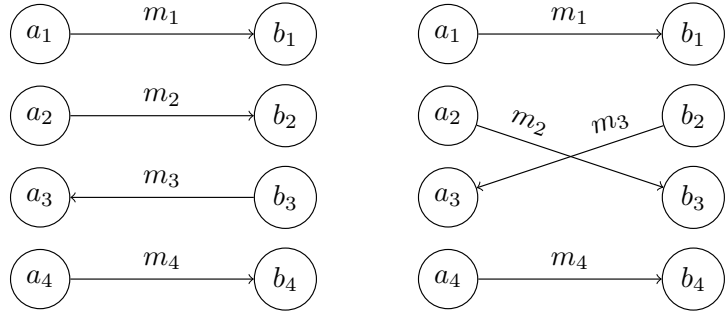

    \centering
    \resizebox{0.8\textwidth}{!}{
        \tikz \graph [math nodes, nodes={circle, draw}, no placement] {
            a_1[x=0 cm,y=3 cm] ->["$m_1$"] b_1[x=3 cm,y=3 cm];
            a_2[x=0 cm, y=2 cm] ->["$m_2$"] b_2[x=3 cm, y=2 cm];
            a_3[x=0 cm,y=1 cm] <-["$m_3$"] b_3[x=3 cm, y=1 cm];
            a_4[x=0 cm,y=0 cm] ->["$m_4$"] b_4[x=3 cm, y=0 cm];
        };
        \hspace{1cm}
        \tikz \graph [math nodes, nodes={circle, draw}, no placement] {
            a_1[x=0 cm,y=3 cm] ->["$m_1$"] b_1[x=3 cm,y=3 cm];
            a_2[x=0 cm, y=2 cm] ->[near start, sloped, above, "$m_2$"] b_3[x=3 cm, y=1 cm];
            b_2[x=3 cm, y=2 cm] ->[near start, sloped, above, "$m_3$"] a_3[x=0
            cm,y=1 cm]; a_4[x=0 cm,y=0 cm] ->["$m_4$"] b_4[x=3 cm, y=0 cm]; }; }
            \caption{Examples of causality graphs. Let the $a$ nodes correspond
            to Alice and the $b$ nodes correspond to Bob. The left graph
            corresponds to a situation in which both Alice and Bob view the left
            ordering from \figref{fig:message_order}. The right graph leads to
            Alice viewing the left ordering and Bob viewing the right ordering
            from the same figure. }
   \label{fig:caus-graph}
\end{figure}

The counters $\sendCtr$ and $\recvCtr$ are monotonically increasing counters
over sending and reception event respectively for each party $P \in \bits$. The
addition of a sending event to $V_P$ increments the sending counter for $V_P$
while the addition of a reception event increments the reception counter.
Consider a sending vertex $v \in V_P$. We have that $v.t = S$, there are
$v.\sendCtr - 1$ sending events that precede $v$ and $v.\recvCtr$ reception
events that precede $v$. The nodes in $V_P$ are totally ordered by the
lexicographic ordering over $v$ given by $(v.\sendCtr, v.\recvCtr)$. For a
message $m'$ sent from $P$ to $\barP$, there are vertices $v_s \in V_P$ and $v_r
\in V_\barP$, where $v_s.t = S$, $v_r.t = R$, $v_s.m = v_r.m = m'$, and $(v_s,
v_r) \in E$.

Now that we have introduced the semantics of the data contained within the
graph, we define graph operations associated with sending and receiving
messages, and how they modify the data contained within the graph.

\paragraph{Graph initialization} An empty graph $G$ consists of empty vertex
sets $V_0$ and $V_1$ along with an empty edge set $E$. Each vertex set $V_P$ has
an associated send counter $\sendCtr$ and reception counter $\recvCtr$. We use
the notation $V_P.\sendCtr$ and $V_P.\recvCtr$ to refer to these counters for
party $P$. Initially, $V_P.\sendCtr = V_P.\recvCtr = 0$ for $P \in \bits$. The
counters will be treated as auxiliary state as opposed to part of the graph.
We use the symbol $\varepsilon$ to denote the empty graph.

\paragraph{Addition of send operation} We use the notation $G \gets G +
(\sendOp, P, m)$ to denote the addition of a sending operation for message $m'$
from party $P$. First, we increment $V_P.\sendCtr$ ($V_P.\sendCtr \gets
V_P.\sendCtr + 1$), then create a new vertex $v$ with $v.t = S$, $v.\sendCtr =
V_P.\sendCtr$, $v.\recvCtr = V_P.\recvCtr$, and $v.m = m'$. Finally, we add $v$
to $V_P$. When updating a non-message-inclusive graph, we simply write $G \gets G +
(\sendOp, P)$. 

\paragraph{Addition of reception operation} The addition of a reception for the
message with sending index $i$ from party $\barP$ is denoted by $G \gets G +
(\recvOp, P, i)$. Note that this operation enforces that such a message with
sending index $i$ exists in $G.V_\barP$. First, we increment $V_P.\recvCtr$,
then create a new vertex $v$ with $v.t = R$, $v.\sendCtr = V_P.\sendCtr$,
$v.\recvCtr = V_P.\recvCtr$, $v.m = v_s.m$, where $v_s \in V_\barP$, $v_s.t =
S$, $v_s.\sendCtr = i$ (by construction, there is only one such vertex). We add
$v$ to $V_P$ and add $(v_s, v)$ to $E$.

\paragraph{Graph validity and consistency} A causality graph $G$ is valid if there exists
a sequence of send and reception operations that lead to its construction from
an empty graph. A sub-graph $G'$ is valid if there exists a valid graph $G$ such
that $G' \subseteq G$. Two graphs $G$ and $G'$ are consistent if there exists a
valid graph $G^*$ such that $G \subseteq G^*$ and $G' \subseteq G^*$. We write
$G \approx G'$ to indicate that $G$ and $G'$ are consistent, and we use $G
\not\approx G'$ to indicate that they are inconsistent.
Notions of validity and consistency will be important in our security
definitions.

\paragraph{Partial ordering over events} As we mentioned before, there is a
total ordering over the events within $V_0$ and $V_1$. For $v_1, v_2 \in V_P$,
we have that $v_1 < v_2$ if $v_1.\sendCtr \leq v_2.\sendCtr$, $v_1.\recvCtr \leq
v_2.\recvCtr$, and $(v_1.\sendCtr, v_1.\recvCtr) \neq (v_2.\sendCtr,
v_2.\recvCtr)$. The causality partial ordering $v_1 \leq v_2$ for $v_1, v_2 \in
V$ is defined as the transitive closure of these individual total orders along
with the edge relation (i.e., $v_1' \leq_E v_2'$ if $(v_1', v_2') \in E$).
Observe that this is an ordering over the sending and reception of messages as
opposed to an ordering of the messages themselves.

\paragraph{Contiguity and gaps}  When a moderator views a sub-graph of a full
conversation, it will be useful to understand which events are contiguous (in
the sense that no other events can be ordered between them) and which events
have gaps. This makes clear where there might be missing context, which can
later be provided by either party if requested. The inclusion of sending and
reception counters within the graph provides rich information that allows
interpretation of the contiguity of events. Consider two vertices $v, v' \in
V_P$ where $v < v'$. If we have that $\max(v'.\sendCtr - v.\sendCtr, v'.\recvCtr
- v.\recvCtr) = 1$, then by construction there can be no $v'' \in V_P$ such that
$v < v'' < v'$, so $v$ and $v'$ are contiguous with one another, among the
vertices in $V_P$. Taking into consideration the edges in $E$, we can ascertain
whether any nodes in $V_\barP$ can be ordered between them. If we have that
$(v'.\sendCtr - v.\sendCtr, v'.\recvCtr - v.\recvCtr) = (a, b)$ for $v, v' \in
V_P$, we know that there are $a$ send events and $b$ reception events that
occurred after $v$, including $v'$. In this way, these counters provide rich
interpretability of gaps within sub-graphs.

\section{Two-Party Transcript Franking}
\label{sec:setting}

Prior work on message franking has focused on reporting single messages that
were received by the reporting user. Often, single messages do not contain
sufficient context for understanding online harassment. Hence, we aim to extend
message franking to enable reporting sequences of messages. We propose a new
cryptographic primitive we call transcript franking that captures this goal. In
addition to providing integrity over the contents of messages contained within a
sequence, a transcript franking scheme must have cryptographic assurance over
the ordering and contiguity of messages reported within a sequence. Due to the
inherent asynchrony of messaging, honest users may observe differing but valid
views of the message order. We aim for users and platform moderators to be able
to see the view of the transcript from the perspective of each party as well the
causal dependencies between messages.

\paragraph{Limitations of current approaches} The original Facebook white-paper
that introduced message franking suggests including server timestamps within the
data tagged by the server \cite{fbSecretConv}. As prior work points out, this is
insufficient for capturing causal dependencies between messages since it does
not directly reflect client-side views and capture message concurrency
\cite{causalityPres}. Furthermore, timestamps do not assure integrity over the
contiguity of reported messages within a conversation. Recent work suggests
incorporating causality metadata into message franking channels, however since
this metadata is client-generated, it can deviate from the actual ordering of
sending and reception events that clients experience \cite{causalityPres}. As we
explain in the introduction, this leads to attacks in which a malicious party
can force the other party to view an upsetting message sequence while only being
able to report a plausibly benign one.

A further limitation is that all prior treatments of message franking allow
reporters to disclose only messages they have received from the other party.
From a motivational standpoint, this seems reasonable since the goal of
reporting is to demonstrate that an abusive party, presumably the non-reporting
party, sent harmful content. However, when reporting multiple messages within a
conversation, a reporter may have to include their own sent messages to provide
sufficient context. In \figref{fig:message_order}, for instance, Alice should be
able to show her message that precedes the message Bob sends in order for the
moderator to evaluate whether it is abusive. The obvious solution of a reporter
disclosing their own sent messages is insufficient since the corresponding
ciphertexts may be malformed for the other party, and existing treatments of
message franking provide no way for recipients to indicate this to the service
provider. An alternative would be to involve both parties in the reporting
process. Yet this can lead to a denial of service attack in which the accused
party refuses to participate. Even with an honest accused party, requiring both
parties to be online is an unfavorable limitation. Hence, we must devise a
reporting protocol that is not contingent upon the participation of the
non-reporting party. Note that such a protocol may still allow the non-reporting
party to continue to disclose more context if they decide to do so.

\paragraph{Our suggested platform model}
In line with the model used by the messaging layer security (MLS) standard
\cite{mlsRfc}, we conceptualize the platform as providing a delivery service
(DS) for message transmission and an authentication service (AS) for managing
user identities. Additionally, we consider a platform-managed moderation service
(MS) for accepting user reports and taking action against abusive user accounts.
Clients issue a \texttt{Send} request with a payload containing the message
ciphertext to the DS. To receive messages, a client issues a \texttt{Recv}
request, to which the server responds with outstanding message ciphertexts. In a
real platform, clients may be identified via usernames or phone numbers. Our
transcript franking abstraction, on the other hand, simply identifies parties
within a conversation via the numeric labels $\bits$ for direct messaging.

Our model assumes a client receives a notification that indicates when the DS
has successfully received their message and another notification when the
recipient client device has successfully received the message. These two events
correspond to the first and second checkmarks shown in widely used messaging
platforms such as Whatsapp and Signal
\cite{whatsappCheckmarks,signalCheckmarks}. Messages may still be dropped or
reordered, but we will concern ourselves with reporting only messages that were
successfully received. We remark that this differs slightly from traditional
platform models for message franking, in which the platform simply tags a
message once it is sent and shares this tag with just the recipient. It turns
out this model will be crucial to achieving our transcript franking security
goals.

\paragraph{Transcript franking syntax and semantics}
We draw heavily on \cite{causalityPres} as inspiration for our syntax and will
later provide a comprehensive comparison between their approach to incorporating
causality into message franking and our notion of transcript franking. A
transcript franking scheme is a tuple of algorithms $\transcriptFrank =
(\srvInit,\allowbreak \Init,\allowbreak \Snd,\allowbreak \Rcv,\allowbreak
\tagSend,\allowbreak \tagRecv,\allowbreak \verifyReport)$, defined over a server
state space $\stSp_S$, a client state space $\stSp_C$, a key space $\keySp$, a
message space $\msgSp$, a commitment space $\commSp$, a franking key space
$\keySp_f$, a message-sent tag space $\tagSp_S$, and a reception tag space
$\tagSp_R$. We detail each algorithm below. Input and output variables names are
unique, and we indicate the ``type'' (the set of possible values) and
description of a variable only the first time it is introduced in order to
reduce verbosity. In general, $\st$ will refer to client or server state the before the
invocation of a routine, while $\st'$ refers to the updated state after the
invocation.

\begin{itemize}
    \item $\servSt \getsr \srvInit()$ outputs an initial server state $\servSt
    \in \stSp_S$.
    \item $\st \getsr \clInit(P, k)$ outputs initial client state $\st \in
          \stSp_C$ for party $\party \in \bits$ and a secret shared channel key
          $k \in \keySp$.
    \item $\st', c \getsr \Snd(P, \st, m)$ is a client procedure that produces a
          ciphertext $c \in \ctSp$ and updated client state $\st' \in \stSp_C$
          for party $P$, with original client state $\st \in \stSp_C$,
          corresponding to an input message $m \in \msgSp$. The ciphertext $c$
          contains a franking tag $c_f \in \commSp$, which is a commitment to
          $m$, and a sending index $i \in \Z^*$.
    \item $\servSt', \sendTag \gets \tagSend(\servSt, P, c_f)$ is a server
          procedure that produces a tag $\sendTag \in \tagSp_S$ for a message
          sending event, where $P$ is the sending party, and $c_f$ is the
          franking tag for the message.
    \item $\servSt, \recvTag \gets \tagRecv(\servSt, P, c_f)$ is a server
          procedure that produces a tag $\recvTag \in \tagSp_R$ for a message
          reception event by receiving party $P$. This procedure is
          invoked only when the receiving client indicates that the message was
          successfully received and valid.
    \item $\st', m, k_f, i \gets \Rcv(\party, \st, c)$ is a client procedure
          that processes a received ciphertext $c \in \ctSp$ and decrypts it to
          a message $m \in \msgSp \cup \setelems{\bot}$ with sending index $i \in
          \Z^*$ and a franking key $k_f \in\keySp_f$. The message $m$ is $\bot$
          if and only if decryption fails. 
    \item $G \gets \verifyReport(\st_S, \reportInfo)$ takes as input the server
          state $\st_S$ as well as a client-provided report $\reportInfo \in
          \{(\bits \times \msgSp \times \keySp_f \times \commSp \times \tagSp_S
          \times \tagSp_R)\}$, which is a set of tuples of the form $(P, m, k_f,
          c_f, \sendTag, \recvTag)$. This procedure verifies the report and, if
          the report is valid, produces a causality graph $G \in (\vertexSpace
          \times \edgeSpace) \cup \setelems{\bot}$ for the messages contained
          within the report. If the reporting information is invalid, the
          procedure outputs $\bot$. Note, we enforce here that only messages
          that have been sent and received can be reported.
\end{itemize}
To illustrate the semantics of a transcript franking scheme, we provide a brief
example of how these procedures are invoked. At initialization time, the server
runs $\srvInit$ to produce an initial state. When two clients initiate a
conversation, both clients run $\init(P, k)$ over a shared key $k$ to obtain
initial client states. When client $P$ sends a message $m$, it obtains $c$ from
the output of $\Snd(P, \st, m)$. The ciphertext $c$ is sent to a platform
server, which then invokes $\tagSend(\servSt, P, c_f)$ to produce a tag
$\sendTag$, which is returned to $P$. Eventually, $\barP$ contacts the server to
retrieve outstanding messages. Upon doing so, the server transmits the
ciphertext $c$, along with $\sendTag$, to $\barP$, which decrypts it by invoking
$\Rcv$. Upon indicating valid reception to the platform, the server runs
$\tagRecv$ to produce $\recvTag$, which is sent to both $P$ and $\barP$. When
party $P$ wishes to report a set of messages, they compile $(P, m, k_f, c_f,
\sendTag, \recvTag)$ for each message in $\reportInfo$. The moderator then runs
$\verifyReport$ to produce a causality graph $G$.

\paragraph{Correctness}
Informally, correctness requires that all honestly generated and tagged messages
can be successfully received and that any sub-graph of the full causality graph
can be reported. We make this precise with the game shown in
\figref{fig:correctness}, and define the correctness advantage of an adversary
$\advA$ as $$\advantageCorr_\tfScheme(\advA) = \Pr[\gameCorr_{\tfScheme}(\advA) = 1]\;.$$
Formally, a transcript franking scheme $\tfScheme$ achieves correctness if
$\advantageCorr_\tfScheme(\advA) = 0$ for all adversaries $\advA$.
\begin{figure*}[t]
    \centering
    \fhpages{0.4}{
        \procedurev{$\gameCorr_{\tfScheme}(\advA)$}\\
        $\serverKey \getsr \keySp$,
        $\st_\advA, \channelKey \getsr \advA()$,
        $\winFlag \gets 0$\\
        $\st_S \getsr \srvInit(\serverKey)$\\
        $\st_0 \getsr \clInit(0, \channelKey)$,
        $\st_1 \getsr \clInit(1, \channelKey)$\\
        $\R_t, \recvMessages, \R \gets \setelems{}, \setelems{}, \setelems{}$\\
        $\advA^\calO(\st_\advA, \channelKey)$\\
        $\creturn \winFlag$\\

        \procedurev{$\calO.\SendTagOracle(P, m)$}
        $(\st_P, c) \getsr \Snd(P, \st_P, m)$\\
        $\st_S, \sendTag \gets \tagSend(\st_S, P, c.c_f)$\\
        $G \gets G + (\texttt{S}, P, m)$\\
        Add $(P, c, \sendTag)$ to $\R_t$\\
        $\creturn c, \sendTag$
    }{

        \procedurev{$\calO.\RecvTagOracle(P, c, \sendTag)$}\\
        Assert $(\barP, c.c_f, \sendTag) \in \R_t$\\
        Assert $(P, c, \sendTag) \not\in \R$\\
        Add $(P, c, \sendTag)$ to $\R$\\
        $\st_P, m, k_f, i \gets \Rcv(P, \st_P, c)$\\
        $\cif m \neq \bot \cthen$\\
        \ind$\st_S, \recvTag \gets \tagRecv(\st_S, P, c_f)$\\
        \ind $G \gets G + (\texttt{R}, P, c.i)$\\
        \ind Add $(P, m, k_f, c_f, t_s, t_r)$ to $\recvMessages$\\
        $\celse$\\
        \ind $\winFlag \gets 1$\\
        $\creturn m, k_f, \sendTag, \recvTag$\\

        \procedurev{$\calO.\ReportOracle(\reportInfo)$}\\
        Assert $|\reportInfo| > 0$\\
        $G' = \verifyReport(\servSt, \reportInfo)$\\
        $\cif \reportInfo \subseteq \R_r \land ((G' = \bot) \lor (G' \not\subseteq
        G))$:\\
        \ind $\winFlag \gets 1$
    }
    \caption{The security game for transcript franking correctness.
    }
    \label{fig:correctness}
\end{figure*}

\paragraph{Tagging reception events}
We discuss in detail what it means for the platform to tag a successful
reception event. Recall that our goal is to enable reporters to include their
own sent messages within a report without interaction from the other party. The
reception tag serves as a way to achieve this goal, acting as an acknowledgement
from the other party that the reported message was received. However, in order
for this acknowledgement to be meaningful, we must carefully consider at which
point the platform generates the reception tag. One option is to generate the
tag once the recipient sends a message to the service provider indicating that
their verification check passed, meaning the franking tag corresponds to the
received plaintext. This would require two round-trips, the first for to
retrieve the message, and the second to explicitly tell the server that it was
well-formed.

Another option is to do this in one round trip: immediately tag the reception
event and send the reception tag along with the ciphertext. If the ciphertext is
malformed, the recipient can issue a complaint to the service provider,
nullifying the reception tag in question. Therefore, two round trips are made
only if the ciphertext is malformed. An implementation may also enforce a
reasonable time window within which to make such a complaint. We discuss
receiver acknowledgement further in \secref{sec:discussion}.

\paragraph{Comparison to Chen-Fischlin} Observe that our formalism, unlike that
of \cite{causalityPres} includes two server-side tagging procedures as opposed
to one. This makes possible acknowledgement of message receipt by the platform
and message delivery to the recipient client device. As a result, the server, as
opposed to client devices, becomes the authority on message ordering, leading to
additional security benefits as we discuss next. Instead of having a single
$\Init$ procedure shared by the client and server, we specify $\srvInit$, and
$\clInit$. The syntax and semantics of the message franking channel in
Chen-Fischlin does not enable parties to report their own sent messages while
our formalism does. While Chen and Fischlin focus on two-party channels, we show
how to enable transcript franking for group channels in \secref{sec:groupmsg}.
We provide a more in-depth comparison in \appref{app:comparison}.

\section{Security Definitions for Two-Party Transcript Franking}
\label{sec:two-party}
In this section, we introduce security notions for transcript franking. Our
setting requires that the platform is the same entity that handles moderation
reports. Recall that we defined the transcript franking syntax and semantics in
\secref{sec:setting}. Our security definitions formalize notions of
confidentiality and accountability for the reporting process. Accountability
consists of two properties, reportability and transcript integrity, which we
further explain in the remainder of the section.

\paragraph{Threat model, informally} As the platform is trusted for handling user reports,
we trust it to serve as source of ground truth for the ordering of messages.
This does not mean that we trust the platform with the contents of messages or
that we assume a malicious platform will not attempt to maul ciphertexts. We do
not explicitly model the public key infrastructure used to authenticate users,
though we note that solutions such as key transparency
\cite{coniks,rzks,optiks,parakeet,versa} enable PKIs without placing full trust in the
service provider.

\paragraph{Transcript reportability}
When a client accepts a message as valid, it should be the case that this
message can be successfully reported to the moderator as well. Transcript
reportability, which is formally specified by a security game in
\figref{fig:reportability}, is a security property we define that captures this
goal. The adversary $\advA$ attempts to craft a report, containing messages
accepted by an honest recipient, that does not verify for the moderator. We
define the reportability advantage of an adversary $\advA$ for a transcript
franking scheme $\tfScheme$ as follows:
$$\advantageRep_{\tfScheme}(\advA) = \Pr[\gameRep_\tfScheme(\advA) = 1]\;.$$

\paragraph{Transcript integrity}
To model malicious reporters that attempt to trick moderators into accepting
incorrect causality graphs, we define a security notion called transcript
integrity, which is captured by the game in \figref{fig:tr-integ}. The adversary
$\advA$ controls both parties and has access to three oracles: $\SendTagOracle$,
$\RecvTagOracle$ and $\ReportOracle$. A ground truth causality graph $G$ is
maintained by the game and updated by $\SendTagOracle$ and $\RecvTagOracle$. The
adversary wins if it can produce two valid reports, where at least one is not a
sub-graph of the ground truth causality graph, or where the generated sub-graphs
are inconsistent.

To elaborate, we recall some definitions regarding causality graphs that were
given in \secref{sec:prelim}. First, recall that $\tilde{G}$ refers to the graph
with message labels removed, allowing us to compare with the ground-truth
causality graph $G$ maintained by the game. Second, two causal sub-graphs are
consistent if there exists a valid causality graph $G'$ of which they are both
sub-graphs. This final consistency condition means that there is a unique ground
truth causality graph from which sub-graphs can be reported. We view this as a
natural lifting of the receiver binding notion proposed in \cite{smfComm} to the
multi-message setting. The advantage of an adversary $\advA$ in the transcript
integrity game is defined as follows:
$$\advantageTrInt_{\tfScheme}(\advA) = \Pr[\gameTrInt_{\tfScheme}(\advA) = 1]\;.$$

\paragraph{Confidentiality}
In order for a transcript franking scheme to achieve confidentiality, the
reporting process must not reveal any information about unreported messages. Of
course, the underlying messaging channel itself must provide confidentiality as
well. We formalize this property in a security game in \figref{fig:conf},
inspired by the real-or-random multi-opening confidentiality notion proposed in
\cite{smfComm}. Our definition uses the function $\clen : \msgSp \to \Z^*$,
which outputs the length of a ciphertext for plaintext $m$.
The ROR-advantage against the confidentiality of a transcript
franking scheme $\tfScheme$ for an adversary $\advA$ is:
$$\advantageConf_{\tfScheme}(\advA) = |\Pr[\gameConf_{\tfScheme, 0}(\advA)] -
      \Pr[\gameConf_{\tfScheme, 1}(\advA)]|\;.$$

\begin{figure*}[t]
    \centering
    \fhpagesss{0.23}{0.36}{0.33}{
        \procedurev{$\gameRep_{\tfScheme}(\advA)$}\\
        $\serverKey \getsr \keySp$\\
        $\st_\advA, \channelKey \getsr \advA()$\\
        $\winFlag \gets 0$; $\R \gets \setelems{}$\\
        $\st_S \getsr \srvInit(\serverKey)$\\
        $\st_0 \getsr \clInit(0, \channelKey)$\\
        $\st_1 \getsr \clInit(1, \channelKey)$\\
        $G, \R_t, \recvMessages \gets \varepsilon, \setelems{}, \setelems{}$\\
        $\advA^\calO(\st_\advA, \channelKey)$\\
        $\creturn \winFlag$
    }{
        \procedurev{$\calO.\RecvTagOracle(P, c, \sendTag)$}\\
        Assert $(\barP, c.c_f, \sendTag) \in \R_t$\\
        Assert $(P, c, \sendTag) \not\in \R$\\
        Add $(P, c, \sendTag)$ to $\R$\\
        $\st_P, m, k_f, i \gets \Rcv(P, \st_P, c)$\\
        $\cif m \neq \bot \cthen$\\
        \ind$\st_S, \recvTag \getsr \tagRecv(\st_S, P, c_f)$\\
        \ind Add $(P, m, k_f, c_f, t_s, t_r)$ to $\recvMessages$\\
        $G \gets G + (\texttt{R}, P, c.i)$\\
        $\creturn m, k_f, \sendTag, \recvTag$\\

        \procedurev{$\calO.\SendOracle(P, m)$}\\
        $(\st_P, c) \getsr \Snd(P, \st_P, m)$\\
        $\creturn c$

    }{
        \procedurev{$\calO.\TagSendOracle(P, c_f)$}\\
        $G \gets G + (\texttt{S}, P)$\\
        $\st_S, \sendTag \getsr \tagSend(\st_S, P, c_f)$\\
        Add $(P, c_f, \sendTag)$ to $\R_t$\\
        $\creturn \sendTag$\\

        \procedurev{$\calO.\ReportOracle(\reportInfo)$}\\
        Assert $|\reportInfo| > 0$\\
        $G' \gets \verifyReport(\servSt, \reportInfo)$\\
        $\cif G' \not\subseteq G \land \reportInfo \subseteq \R_r$:\\
        \ind $\winFlag \gets 1$
    }

    \caption{The security game for transcript reportability.
    }
    \label{fig:reportability}
\end{figure*}

\begin{figure*}[t]
    \centering
    \fhpagesss{0.23}{0.35}{0.35}{
        \procedurev{$\gameTrInt_{\tfScheme}(\advA)$}\\
        $\serverKey \getsr \keySp$; $\winFlag \gets 0$\\
        $\st_\advA, \channelKey \getsr \advA()$\\
        $\st_S \getsr \srvInit(\serverKey)$\\
        $\st_0 \getsr \clInit(0, \channelKey)$\\
        $\st_1 \getsr \clInit(1, \channelKey)$\\
        $G, \R_t, \R \gets \varepsilon, \setelems{}, \setelems{}$\\
        $\advA^\calO(\st_\advA, \channelKey)$\\
        $\creturn \winFlag$
    }{
        \procedurev{$\calO.\SendTagOracle(P, c)$}\\
        $G \gets G + (\texttt{S}, P)$\\
        $\st_S, \sendTag \gets \tagSend(\st_S, P, c.c_f)$\\
        Add $(P, c, \sendTag)$ to $\R_t$\\
        $\creturn \sendTag$
    }{
        \procedurev{$\calO.\RecvTagOracle(P, c, \sendTag)$}\\
        Assert $(\barP, c, \sendTag) \in \R_t$\\
        Assert $(P, c, \sendTag) \not\in \R$\\
        Add $(P, c, \sendTag)$ to $\R$\\
        $\st_S, \recvTag \gets \tagRecv(\st_S, P, c.c_f)$\\
        $G \gets G + (\texttt{R}, P, c.i)$\\
        $\creturn \recvTag$\\

        \procedurev{$\calO.\ReportOracle(\reportInfo_1, \reportInfo_2)$}\\
        Assert $|\reportInfo_1| > 0$ and $|\reportInfo_2| > 0$\\
        $G_1 \gets \verifyReport(\servSt, \reportInfo_1)$\\
        $G_2 \gets \verifyReport(\servSt, \reportInfo_2)$\\
        $\cif G_1 \neq \bot \land G_2 \neq \bot \land\\
        \ind ((\widetilde{G_1} \not\subseteq G) \lor (\widetilde{G_2} \not\subseteq G)\\\ind \lor (G_1
        \not\approx G_2))$:\\
        \ind $\winFlag \gets 1$
    }
    \caption{The security game for transcript integrity.}
    \label{fig:tr-integ}
\end{figure*}

\begin{figure*}[t]
    \centering
    \fhpages{0.4}{
        \procedurev{$\gameConf_{\tfScheme, b}(\advA)$}\\
        $\serverKey \getsr \keySp_S$; $\channelKey \getsr \keySp_C$\\
        $\st_\advA \getsr \advA()$,
        $\queriedCt \gets \setelems{}$\\
        $\st_S \getsr \srvInit(\serverKey)$\\
        $\st_0 \getsr \clInit(0, \channelKey)$,
        $\st_1 \getsr \clInit(1, \channelKey)$\\
        $\hat{b} \gets \advA^\calO(\st_\advA)$\\
        $\creturn \hat{b}$\\

        \procedurev{$\calO.\SendOracle(P, m)$}
        $(\st_P, c) \getsr \Snd(P, \st_P, m)$\\
        $\queriedCt \gets \queriedCt \cup \setelems{c}$\\
        $\creturn c$
    }{
        \procedurev{$\calO.\RecvOracle(P, c, \sendTag)$}\\
        Assert $c \in \queriedCt$\\
        $\st_P, m, k_f, i \gets \Rcv(P, \st_P, c)$\\
        $\creturn m, k_f$\\

        \procedurev{$\calO.\ChalSendOracle(P, m)$}\\
        $(\st_P, c_0) \getsr \Snd(P, \st_P, m)$\\
        $c_1 \getsr \bits^{\clen(m)}$\\
        $\creturn c_b$

    }
    \caption{The security game for transcript franking confidentiality.
    }
    \label{fig:conf}
\end{figure*}

\section{Our Construction}\label{sec:construction}

The key idea of our construction is to report platform-tagged acknowledgements
of message sending and reception. These acknowledgements contain counters,
maintained as part of the server state, that allow the moderator to reliably
reconstruct a portion of the causality graph corresponding to the platform's
view of message transmission. Since our construction uses four counters per
conversation, we call it the quad-counter construction (QCC). We present the
pseudocode specification of our construction in \figref{fig:tf-dm-const}, which
specifies how the service provider handles state for a single conversation
between two parties. Parallelizing this for multiple conversations can be done
in a straightforward manner, as we further discuss in
\secref{sec:discussion}.

\paragraph{Client logic} The client procedures $\Init$, $\Snd$, and $\Rcv$
comprise a secure messaging channel with reportable franking tags $c_f$,
committing to plaintext content $m$, with the opening $k_f$. Indeed, these three
procedures form a message franking channel \cite{msgFrChan, causalityPres}.

\paragraph{Server logic} Upon server initialization, sending and reception
counters for each party are initialized to 0 and the server samples a MAC key.
When a party $P$ sends a message, the server increments the send counter for $P$
and tags the send event. Similarly, it increments the reception counter for $P$
when $P$ successfully receives a message, and then produces a tag for this
event. In the pseudocode, the symbols $S$ and $R$ are labels that denote sending
and reception events respectively.

\paragraph{Reporting}
To report a set of messages, a client compiles the message contents $m$, the
franking key $k_f$, the franking tag $c_f$, the send tag $\tagSend$, and the
reception tag $\recvTag$ for each message. The client then forwards this
information to the platform within a single report object $\reportInfo$. The
platform verifies the commitments for each message along with its own MAC tags.
Then, it uses the indexes within the tags to construct and order the vertices
for the sub-graph, and it adds edges between vertices that correspond to the
same message. A moderator can interpret the contiguity of vertices as explained
in \secref{sec:prelim}.

\begin{figure*}[t]
    \centering
    \fhpages{0.45}{
        \procedurev{$\srvInit()$}\\
        $\macKey \getsr \keySp$; $\sendCtr_0, \recvCtr_0, \sendCtr_1, \recvCtr_1 \gets (0, 0, 0, 0)$\\
        $\creturn \{\macKey, \sendCtr_0, \recvCtr_0, \sendCtr_1, \recvCtr_1\}$\\

        \procedurev{$\clInit(P, k)$}\\
        $\creturn \channel.\init(P, k)$\\

        \procedurev{$\Snd(P, \st, m)$}\\
        $(k_f, c_f) \gets \Com(m)$\\
        $(\st.\channelSt, c_e) \getsr \channel.\Snd(\party, \st.\channelSt, (m,
        k_f))$\\
        $\creturn \st, (c_e, c_f)$\\

        \procedurev{$\tagSend(\servSt, P, c_f)$}\\
        $\sendCtr_P \gets \sendCtr_P + 1$, $\ack \gets (S, P, \barP, c_f,
            \sendCtr_P, \recvCtr_P)$\\
        $\sendTag \gets (\ack, \mac(\macKey, \ack))$\\
        $\creturn \st_S, \sendTag$\\

        \procedurev{$\tagRecv(\servSt, P, c_f)$}\\
        $\recvCtr_P \gets \recvCtr_P + 1$, $\ack \gets (R, \barP, P, c_f,
            \sendCtr_P, \recvCtr_P)$\\
        $\recvTag \gets (\ack, \mac(\macKey, \ack))$\\
        $\creturn \st_S, \recvTag$
    }{

        \procedurev{$\Rcv(\party, \st, c, \sendTag, \recvTag)$}\\
        $(\st.\channelSt, m, k_f, i) \gets
            \channel.\Rcv(\party, \channelSt, c)$\\
        $\cif m = \bot \lor \VerC(m, k_f,
            c.c_f) = 0 $:\\
        \ind $\creturn \bot$\\
        $\creturn \st, m, k_f, i$\\

        \procedurev{$\verifyReport(\servSt, \reportInfo)$}\\
        Initialize empty graph $G$\\
        For $(P, m, k_f, c_f, \sendTag, \recvTag)$ in $\reportInfo$\\
        \ind $b \gets \Ver(\macKey, \sendTag.\ack, \sendTag.\tagLabel) \land\\
        \ind\ind \Ver(\macKey, \recvTag.\ack,
        \recvTag.\tagLabel)\land$\\
        \ind\ind$\VerC(m, k_f, c.c_f)\land \sendTag[0] = S \land\\\ind\ind \recvTag[0] =
        R\land \sendTag.c_f = \recvTag.c_f$\\
        \ind $\cif b = 0$:\\
        \ind \ind $\creturn \bot$\\
        \ind $\sendCtr_P, \recvCtr_P = \sendTag.\ack.(\sendCtr, \recvCtr)$\\
        \ind $\sendCtr_\barP, \recvCtr_\barP = \recvTag.\ack.(\sendCtr, \recvCtr)$\\
        \ind $v_s = (S, \sendCtr_P, \recvCtr_P, m)$\\
        \ind $v_r = (R, \sendCtr_\barP, \recvCtr_\barP, m)$\\
        \ind Add $v_s$ to $G.V_P$, add $v_r$ to $G.V_\barP$\\
        \ind Add $(v_s, v_r)$ to $G.E$\\
        $\creturn G$
    }
    \caption{Pseudocode for our two-party transcript franking construction.}
    \label{fig:tf-dm-const}
\end{figure*}

\paragraph{Security proofs}
We now demonstrate the security of our transcript franking construction,
$\tfScheme$. We begin by proving transcript integrity. The following lemma will
be helpful for our proof.
\begin{lemma}
      Let $G_1$ and $G_2$ be two valid two-party causality sub-graphs. Suppose
      $\widetilde{G_1} \approx \widetilde{G_2}$ but $G_1 \not\approx G_2$. Then
      there must be $v_1 \in G_1.V$ and $v_2 \in G_2.V$ such that $v_1[0 : 2] =
      v_2 [0 : 2]$ but $v_1.m \neq v_1.m$.
      \label{lemma:consistency}
\end{lemma}

\begin{proof}
      Assume for the sake of contradiction that for all $v_1 \in G_1.V$ and $v_2
      \in G_2.V$, that if $v_1[0 : 2] = v_2 [0 : 2]$, then $v_1.m = v_2.m$.
      Since $\widetilde{G_1} \approx \widetilde{G_2}$, there exists some valid
      causality graph $G$ such that $\widetilde{G_1}, \widetilde{G_2} \subseteq
      G$. This means that there exists a sequence of send and receive operations
      that constructs $G$. We can use the same sequence of operations to
      generate a valid message-inclusive graph $G^*$, such that $G_1, G_2
      \subseteq G^*$, contradicting the assumption that $G_1 \not\approx G_2$.
      In each send operation, we simply include the message corresponding to the
      vertex $v \in G_1.V \cup G_2.V$, if the counters for that send operation
      correspond to a vertex in the union of the two vertex sets. By our initial
      assumption a unique such vertex exists. For all other send operations, we
      can include an arbitrary message, the empty string for instance.  This
      completes the proof.%
\qed
\end{proof}

\begin{theorem}      \label{thm:dm-tr-int}
  Let $\tfScheme$ be the transcript franking scheme given in
  \figref{fig:tf-dm-const}.
      Let $\advA$ be a transcript integrity adversary against $\tfScheme$. Then we give
      an EUF-CMA adversary $\advB$ and V-Bind adversary
      $\advC$ 
            such that
      \begin{align*}
            \advantageTrInt_{\tfScheme}(\advA) \leq \advantageEUFCMA_{\mathsf{MAC}}(\advB)
             + \advantageVBind_{\mathsf{CS}}(\advC)\;.
      \end{align*}
      Adversaries $\advB$ and $\advC$ run in time that of $\advA$ plus a small
      overhead.
\end{theorem}

\begin{proof}
      We proceed via a sequence of game hops with Game $\game_0$ equivalent to
      $\gameTrInt_{\tfScheme}$, defined in \figref{fig:tr-integ}. To aid with future game
      definitions, we provide some additional bookkeeping, initializing an empty
      set $\R'$ at the start of the game. Game $\game_0$ adds $(S, P, \barP,
      c.c_f, G.(\sendCtr_P, \recvCtr_P))$ to $\R'$ before the $\creturn$
      statement of $\SendTagOracle$. Similarly, it adds $(R, \barP, P, c.c_f,
      G.(\sendCtr_P, \recvCtr_P))$ to $\R'$ before the $\creturn$ statement of
      $\RecvTagOracle$.

      The adversary $\advA$ can only win if it is able to produce
      $\reportInfo_1, \reportInfo_2$ such that $\widetilde{G_1} \not\subseteq G$
      or $\widetilde{G_2} \not\subseteq G$ or $G_1 \not\approx G_2$, where $G_1
      \gets \verifyReport(\servSt, \reportInfo_1)$, $G_2 \gets
      \verifyReport(\servSt, \reportInfo_2)$, and $G$ is the ground truth graph
      maintained by the game. From $(\reportInfo_1, \reportInfo_2)$, we will
      show that we can break either the unforgeability of the MAC or the binding
      property of the commitment.
      
      We consider two cases: (1) the adversary wins with $\widetilde{G_1}$ or
      $\widetilde{G_2}$ not a sub-graph of $G$ or (2) the adversary wins with
      $G_1 \not\approx G_2$, but $\widetilde{G_1}, \widetilde{G_2} \subseteq G$.
      The first case will allow us to reduce to the EUF-CMA security of the MAC
      while the second allows us to reduce to the binding security of the
      commitment. Each case will correspond to a distinct failure event. Let
      $\game_1$ be the same as $\game_0$, except we abort and output 0, right
      before setting the $\winFlag$ flag, if $\advA$ produces a valid
      $(\reportInfo_1, \reportInfo_2)$ in case (1). We denote the event that
      this abort occurs as $F_1$. Let $\game_2$ be the same as $\game_1$ except
      we abort and output 0 at the same location if $\advA$ produces a valid
      $(\reportInfo_1, \reportInfo_2)$ in case (2). We denote the event that
      this abort occurs as $F_2$. Note that $|\Pr[\game_0(\advA) \Rightarrow 1]
      - \Pr[\game_1(\advA) \Rightarrow 1]| \leq \Pr[F_1]$ and
      $|\Pr[\game_1(\advA) \Rightarrow 1] - \Pr[\game_2(\advA) \Rightarrow 1]|
      \leq \Pr[F_2]$. Furthermore, $\Pr[\game_2(\advA) \Rightarrow 1] = 0$, so
      we have $\advantageTrInt(\advA) = \Pr[\game_0(\advA) \Rightarrow 1] \leq
      \Pr[F_1] + \Pr[F_2]$.

      We now demonstrate an adversary $\advB$ where $\advantageEUFCMA(\advB) =
      \Pr[F_1]$. The adversary $\advB$ perfectly simulates $\game_0$ to $\advA$,
      while routing $\mac$ and $\verify$ calls to its challenger oracles. If
      $F_1$ occurs, then we have that $\widetilde{G_i} \not\subseteq G$ for some
      $i \in \setelems{1, 2}$. For $r$ an element of $\reportInfo$, define $f(r)
      = \setelems{r.\sendTag.\ack, r.\recvTag.\ack}$. Observe that
      $\widetilde{G_i} \not\subseteq G$ implies there is some $r^* = (P^*, m^*,
      k_f^*, c_f^*, t_s^*, t_r^*) \in \reportInfo_i$, such that $f(r^*)
      \not\subseteq \R'$.
      
      To see why this is, observe that the construction mirrors the updates of
      the causality graph perfectly. Put formally, if $\bigcup_{r' \in
      \reportInfo} f(r') \subseteq \R'$ and $G' \gets
      \verifyReport(\reportInfo)$, then $\widetilde{G'} \subseteq G$. We have
      that $G.(\sendCtr_P, \recvCtr_P, \sendCtr_\barP, \recvCtr_\barP) = (0, 0,
      0, 0)$ and the server counters $(\sendCtr_P, \recvCtr_P, \sendCtr_\barP,
      \recvCtr_\barP) = (0, 0, 0, 0)$ at the start of the game. When
      $\SendTagOracle$ is called, we increment $G.\sendCtr_P$ and
      $\servSt.\sendCtr_P$. Similarly, when $\RecvTagOracle$ is invoked, we
      increment $G.\recvCtr_P$ and $\servSt.\recvCtr_P$. A simple proof by
      induction on the number of oracle calls shows that $G.(\sendCtr_P,
      \recvCtr_P) = \servSt.(\sendCtr_P, \recvCtr_P)$ for $P \in \bits$ by the
      end of each call to $\SendTagOracle$ and $\RecvTagOracle$. This means that
      $v \in \widetilde{G'}.V$ implies $v \in G.V$ and $e \in \widetilde{G'}.E$
      implies $e \in G.E$.
      
      If $F_1$ occurs, then we retrieve the $r^* =
      (P^*, m^*, k_f^*, c_f^*, t_s^*, t_r^*)$ in question, and observe that
      $\verify(\macKey,\allowbreak t_s^*.\ack,\allowbreak t_s^*.\tagLabel) = 1$
      and $\verify(\macKey,\allowbreak t_r^*.\ack,\allowbreak t_r^*.\tagLabel) =
      1$, because the output of $\verifyReport$ is not $\bot$. We must have that
      at least one of $t_r^*.\ack$ or $t_s^*.\ack$ was not queried to the MAC
      challenger oracle, otherwise both would be in $\R'$. Let $t^*$ denote this
      tag. We output $(t^*.\ack, t^*.\tagLabel)$ as a forgery. 

      Now, we construct adversary $\advC$ where $\advantageVBind_{\commScheme}(\advC) =
      \Pr[F_2]$. If $F_2$ occurs, we have that the adversary $\advA$ was able to
      trigger $G_1 \not\approx G_2$ while $\widetilde{G_1}, \widetilde{G_2}
      \subseteq G$. By \lemmaref{lemma:consistency} there exists $v_1 \in G_1$ and $v_2 \in
      G_2$ such that $v_1[0:2] = v_2[0:2]$ but $v_1.m \neq v_2.m$. Note that
      there must also be a single $c_f$ and $k_f^{(1)}, k_f^{(2)}$ such that
      $\VerC(v_1.m, k_f^{(1)}, c_f) = 1$ and $\VerC(v_2.m, k_f^{(2)}, c_f) = 1$.
      This breaks the binding property of the commitment, and so $\advC$ outputs
      $(v_1.m, k_f^{(1)}, v_2.m, k_f^{(2)}, c_f)$ to win with probability
      $\Pr[F_2]$. This completes the proof.
      \qed
\end{proof}

We now show that our scheme also achieves perfect reportability.
\begin{theorem}
      For all adversaries $\advA$, we have $\advantageRep_{\tfScheme}(\advA) =
      0$.
\end{theorem}
\begin{proof}
      Observe that the check that an honest recipient performs in $\Rcv$,
      namely, that $\VerC(m, k_f, c.c_f) = 0$, is replicated in $\verifyReport$.
      The only way $\verifyReport$ can return $\bot$ is if this check fails, if
      any of the MAC checks fail, or if the input $\reportInfo$ is malformed.
      Since we are dealing with an honest reporter, this cannot be the case, so
      $\verifyReport$ must always return a non-$\bot$ value.
      \qed
\end{proof}

Our scheme achieves correctness via the correctness of the underlying channel
$\channel$, the correctness of the MAC scheme $(\mac, \verify)$, the correctness
of the commitment scheme $(\Com, \VerC)$, and the fact that the counters in our
construction perfectly mirror those of the ground truth graph (see the proof of
\thmref{thm:dm-tr-int}). Observe that our
construction boils down to a commit-then-encrypt scheme, which was proven secure
for the multi-opening real-or-random confidentiality notion in \cite{smfComm},
hence we omit the proof of confidentiality here.

\section{Multi-party Transcript Franking}
\label{sec:groupmsg}
Up to this point, our constructions have considered transcript franking in
the two-party direct messaging context. We now discuss how our approach
generalizes to an arbitrary number of parties.
A group consists of a set of $N$ parties $\{0, \ldots, N-1\}$. The goal
of a group transcript franking construction is to be able to reconstruct a
causality graph like that shown in \figref{fig:group-caus-graph}. Note that,
unlike the two-party setting, one send event can correspond to multiple
reception events, as there are now multiple recipients. Each edge
corresponds to a single copy of the broadcast message. We build on the
multi-party channel communication graph formalism proposed in \cite{mar17},
adapting it to our causality graph abstraction.

\begin{figure}[t]
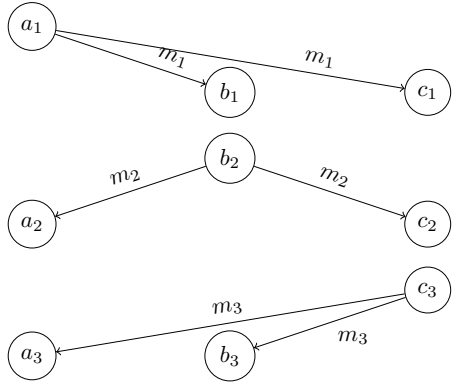

    \centering
    \resizebox{0.5\textwidth}{!}{
        \tikz \graph [math nodes, nodes={circle, draw}, no placement] {
            a_1[x=0 cm,y=5 cm] ->[near end, sloped, above, "$m_1$"] b_1[x=3 cm,y=4 cm];
            a_1 ->[near end, sloped, above, "$m_1$"] c_1[x= 6 cm,y=4 cm];
            b_2[x=3 cm, y=3 cm] ->[sloped, above, "$m_2$"] a_2[x=0 cm,y= 2 cm];
            b_2 ->[sloped, above, "$m_2$"] c_2[x= 6 cm, y=2 cm];
            c_3[x=6 cm, y=1 cm] ->[sloped, above, "$m_3$"] a_3[x= 0 cm, y = 0 cm];
            c_3 ->[above, "$m_3$"] b_3[x=3 cm, y=0 cm];

        };
    }
    \caption{Example of group causality graph }
    \label{fig:group-caus-graph}
\end{figure}

\paragraph{Causality graphs for group messaging}
For $N$-party communication, a causality graph is an $N$-partite graph $G = (V,
E)$, where $V = \bigcup_{i\in[N]} V_i$. All vertex sets $V_i$ for $i \in [N]$
are disjoint. The edge set $E$ consists of pairs $(v, v')$ where $v \in V_i$ and
$v' \in V_j$, where $i, j \in [N]$ and $i \neq j$. The vertex space, as before,
is $\setelems{S, R} \times \Z^{*} \times \Z^{*} \times \msgSp$. The notation
and updates for adding a send event is the same as the two party version. For
adding a reception event to party $P_R$ from party $P_S$, we write $G \gets G +
(\recvOp, P_S, P_R, c.i)$, incrementing the reception counter for $P_R$ before
adding the vertex. We then add an edge between the sending vertex and the new
reception vertex, as before. The partial ordering over events is given by the
transitive closure over the total orders for each vertex set and the edge
relation. The total order within each vertex set is given by the lexicographic
ordering over $v.(\sendCtr, \recvCtr)$ for $v\in V_P$. Event contiguity and gaps
can be interpreted similarly as the two-party case, as described in
\secref{sec:prelim}.

\paragraph{Group messaging channels}
A group messaging channel $\channel$ is defined as the tuple $\channel = (\init,
\Snd, \Rcv)$. The main difference here is that instead of $P \in \bits$, we have
that $P \in [N]$, and our reception procedure accepts the identity of the
sending party $P_S$ associated with the ciphertext $c$.
\begin{itemize}
    \item $\st \getsr \clInit(P, k)$ outputs initial client state $\st \in
          \stSp_C$ for a new channel for Party $\party \in [N]$ and a key $k \in
          \keySp$.
    \item $\st', c \getsr \Snd(P, \st, m)$ is a client procedure that
    produces a ciphertext $c \in \ctSp$ corresponding to an input message $m\in \msgSp$, and
    an updated client state $\st' \in \stSp_C$.
    \item $\st', m, i \gets \Rcv(P_R, \st, P_S, c)$ is a client procedure that
          processes a received ciphertext $c$ from party $P_S \in [N]$ to party
          $P_S \in [N]$ (where $P_S \neq P_R$) and decrypts it to a message $m
          \in \msgSp \cup \setelems{\bot}$
          with sending index $i \in \Z^*$. The message $m$ can be $\bot$ if decryption
          fails.
\end{itemize}

\paragraph{Group transcript franking syntax and semantics}
The $\init$, $\Snd$, and $\Rcv$ procedures inherit the syntactic changes
discussed above, and our $\tagRecv$ procedure accepts an additional argument for
the sending party $P_S \in [N]$. This additional argument is required in the
group case since a client can receive a message from more than one party.
We inherit notational conventions from \secref{sec:two-party}.

\begin{itemize}
    \item $\servSt \getsr \srvInit(N)$ outputs initial server state $\servSt
    \in \stSp_S$ for an $N$-party group.
    \item $\st \getsr \clInit(P, k)$ outputs initial client state $\st \in \stSp_C$ for a new
          channel for Party $\party \in [N]$ and a key $k \in \keySp$.
    \item $\st', c \getsr \Snd(P, \st, m)$ is a client procedure that
          produces a ciphertext $c \in \ctSp$ corresponding to an input message
          $m \in \msgSp$, and
          an updated client state $\st' \in \stSp_C$.
    \item $\servSt', \sendTag \gets \tagSend(\servSt, P_S, c_f)$ is a
          server procedure that produces a tag $\sendTag \in \tagSp_S$ for a message
          sending event, where $P_S \in [N]$ is the sending party, $c_f \in \commSp$ is the
          franking tag for the message, and $\servSt' \in \stSp_S$ is the updated server
          state.
    \item $\servSt, \recvTag \gets \tagRecv(\servSt, P_S, P_R, c_f)$ is a server
          procedure that produces a tag $\recvTag\in \tagSp_R$ for a message
          reception event by receiving party $P_R \in [N]$ for a message sent by
          party $P_S \in [N]$. This procedure is invoked only when the receiving
          client indicates that the message was successfully received and valid.
    \item $\st', m, k_f, i \gets \Rcv(P_R, \st, P_S, c)$ is a client procedure
          that processes a received ciphertext $c \in \ctSp$ and decrypts it to
          a message $m \in \msgSp \cup \setelems{\bot}$ and a franking key $k_f
          \in \keySp_f$. The message $m$ can be $\bot$ if decryption fails.
    \item $G \gets \verifyReport(\st_S, \reportInfo)$ takes as input the server
          state $\st_S \in \stSp_S$ as well as a client-provided report
          $\reportInfo$, which is a set of tuples of the form $(P_S, P_R, m,
          k_f, c_f, \sendTag, \recvTag)$. This procedure verifies the report
          and, if the report is valid, produces a causality graph $G \in
          (\vertexSpace \times \edgeSpace) \cup \setelems{\bot}$ for the
          messages contained within the report. If the reporting information is
          invalid, the procedure outputs $\bot$.
\end{itemize}

\paragraph{Security definitions}
Our security notions in the group setting are a natural extension of those
for the two-party setting.
The main difference is that $N$ parties are initialized and any of these parties
can send and receive messages within the same channel. Correctness and
confidentiality definitions generalize in a straightforward manner, hence we
omit full descriptions of them for brevity. We present our group transcript
reportability definition in \figref{fig:group-reportability} and our group
transcript integrity definition in \figref{fig:group-tr-integ}. To denote the
advantage of an adversary $\advA$ in the $N$-party reportability game, we write
$\advantageRep_{\tfScheme, N}(\advA)$. Similarly, $\advantageTrInt_{\tfScheme,
N}(\advA)$ is the advantage of $\advA$ in the $N$-party transcript integrity
game.

\begin{figure*}[t]
    \centering
    \fhpagesss{0.23}{0.37}{0.33}{
        \procedurev{$\gameRep_{\tfScheme, N}(\advA)$}\\
        $\serverKey \getsr \keySp$\\
        $\st_\advA, \channelKey \getsr \advA()$\\
        $\winFlag \gets 0$; $\R \gets \setelems{}$\\
        $\st_S \getsr \srvInit(\serverKey)$\\
        For $i \in [N]$\\
        \ind $\st_i \getsr \clInit(i, \channelKey)$\\
        $\R_t, \recvMessages \gets \setelems{}, \setelems{}$\\
        $\advA^\calO(\st_\advA, \channelKey)$\\
        $\creturn \winFlag$
    }{

        \procedurev{$\calO.\RecvTagOracle(P_R, P_S, c, \sendTag)$}\\
        Assert $(P_S, c.c_f, \sendTag) \in \R_t$\\
        Assert $(P, c, \sendTag) \not\in \R$\\
        Add $(P, c, \sendTag)$ to $\R$\\
        $\st_P, m, k_f, i \gets \Rcv(P, \st_P, c)$\\
        $\cif m \neq \bot \cthen$\\
        \ind$\st_S, \recvTag \getsr \tagRecv(\st_S, P_R, P_S, c_f)$\\
        \ind Add $(P_S, P_R, m, k_f, c_f, t_s, t_r)$\\\ind\ind to $\recvMessages$\\
        $\creturn m, k_f, \sendTag, \recvTag$\\

    }{
        \procedurev{$\calO.\TagSendOracle(P, c_f)$}\\
        $\st_S, \sendTag \gets \tagSend(\st_S, P, c_f)$\\
        Add $(P, c_f, \sendTag)$ to $\R_t$\\
        $\creturn \sendTag$\\

        \procedurev{$\calO.\ReportOracle(\reportInfo)$}\\
        Assert $|\reportInfo| > 0$\\
        $G' = \verifyReport(\servSt, \reportInfo)$\\
        $\cif G' = \bot \land \reportInfo \subseteq \R_r$:\\
        \ind $\winFlag \gets 1$
    }

    \caption{The security game for transcript reportability for $N$-party
    messaging.
    }
    \label{fig:group-reportability}
\end{figure*}

\begin{figure*}[t]
    \centering
    \fhpagesss{0.23}{0.31}{0.38}{
        \procedurev{$\gameTrInt_{\tfScheme, N}(\advA)$}\\
        $\serverKey \getsr \keySp$; $\winFlag \gets 0$\\
        $\st_\advA, \channelKey \getsr \advA()$\\
        $\st_S \getsr \srvInit(\serverKey)$\\
        For $i \in [N]$\\
        \ind $\st_i \getsr \clInit(i, \channelKey)$\\
        $G, \R_t, \R \gets \varepsilon, \setelems{}, \setelems{}$\\
        $\advA^\calO(\st_\advA, \channelKey)$\\
        $\creturn \winFlag$
    }{
        \procedurev{$\calO.\SendTagOracle(P, m, c, k_f)$}\\
        $G \gets G + (\texttt{S}, P, m)$\\
        $\st_S, \sendTag \gets \tagSend(\st_S, P, c.c_f)$\\
        Add $(P, c, \sendTag)$ to $\R_t$\\
        $\creturn \sendTag$
    }{
        \procedurev{$\calO.\RecvTagOracle(P_R, P_S, c, \sendTag)$}\\
        Assert $(\barP, c, \sendTag) \in \R_t$\\
        Assert $(P, c, \sendTag) \not\in \R$\\
        Add $(P, c, \sendTag)$ to $\R$\\
        $\st_S, \recvTag \gets \tagRecv(\st_S, P_R, P_S, c.c_f)$\\
        $G \gets G + (\texttt{R}, P_S, P_R, c.i)$\\
        $\creturn \recvTag$\\

        \procedurev{$\calO.\ReportOracle(\reportInfo_1, \reportInfo_2)$}\\
        Assert $|\reportInfo_1| > 0$ and $|\reportInfo_2| > 0$\\
        $G_1 \gets \verifyReport(\servSt, \reportInfo_1)$\\
        $G_2 \gets \verifyReport(\servSt, \reportInfo_2)$\\
        $\cif G_1 \neq \bot \land G_2 \neq \bot \land\\
        \ind ((G_1 \not\subseteq G) \lor (G_2 \not\subseteq G)\\\ind \lor (G_1
        \not\approx G_2))$:\\
        \ind $\winFlag \gets 1$
    }
    \caption{The security game for group transcript integrity for $N$-party messaging.}
    \label{fig:group-tr-integ}
\end{figure*}

\paragraph{Our construction}
The group messaging transcript franking construction generalizes naturally from
the two-party version. We provide a pseudocode specification of our group
transcript franking protocol in \figref{fig:tf-gm-const}. Counter updates happen
nearly identically in $\tagSend$ and $\tagRecv$, except now $N$ pairs of
counters are maintained, one for each party.

\begin{figure*}[t]
    \centering
    \fhpages{0.45}{
        \procedurev{$\srvInit(N)$}\\
        $\macKey \getsr \keySp$\\
        For $i\in[N]$ \textbf{do} $\sendCtr_i, \recvCtr_i \gets 0, 0$\\
        $\creturn \{\macKey\} \cup \setelems{\sendCtr_i, \recvCtr_i}_{i\in[N]}$\\

        \procedurev{$\clInit(P, k)$}\\
        $\creturn \channel.\init(P, k)$\\

        \procedurev{$\Snd(P, \st, m)$}\\
        $(k_f, c_f) \gets \Com(m)$\\
        $(\st.\channelSt, c_e) \getsr \channel.\Snd(\party,
        \st.\channelSt, (m, k_f))$\\
        $\creturn \st, (c_e, c_f)$\\

        \procedurev{$\tagSend(\servSt, P, c_f)$}\\
        $\sendCtr_P \gets \sendCtr_P + 1$, $\ack \gets (S, P, c_f,
            \sendCtr_P, \recvCtr_P)$\\
        $\sendTag \gets (\ack, \mac(\macKey, \ack))$\\
        $\creturn \st_S, \sendTag$\\

        \procedurev{$\tagRecv(\servSt, P_R, P_S, c_f)$}\\
        $\recvCtr_{P_R} \gets \recvCtr_{P_R} + 1$\\
        $\ack \gets (R, P_S, P_R, c_f,
        \sendCtr_{P_R}, \recvCtr_{P_R})$\\
        $\recvTag \gets (\ack, \mac(\macKey, \ack))$\\
        $\creturn \st_S, \recvTag$
    }{

        \procedurev{$\Rcv(P_R, \st, P_S, c)$}\\
        $(\st.\channelSt, m, k_f, i) \gets
            \channel.\Rcv(P_R, \channelSt, P_S, c)$\\
        $\cif m = \bot \lor \VerC(m, k_f,
            c.c_f) = 0 $:\\
        \ind $\creturn \bot$\\
        $\creturn \st, m, k_f, i$\\

        \procedurev{$\verifyReport(\servSt, \reportInfo)$}\\
        Initialize empty graph $G$\\
        For $(P_S, P_R, m, k_f, c_f, \sendTag, \recvTag)$ in $\reportInfo$:\\
        \ind $b \gets \Ver(\macKey, \sendTag.\ack, \sendTag.\tagLabel) \land\\
        \ind\ind \Ver(\macKey, \recvTag.\ack,
        \recvTag.\tagLabel)\land$\\
        \ind\ind$\VerC(m, k_f, c.c_f)\land\\\ind\ind \sendTag[0] = S \land \recvTag[0] =
        R\land$\\\ind\ind $\sendTag.c_f = \recvTag.c_f$\\
        \ind $\cif b = 0$:\\
        \ind \ind $\creturn \bot$\\
        \ind $\sendCtr_P, \recvCtr_P = \sendTag.\ack.(\sendCtr, \recvCtr)$\\
        \ind $\sendCtr_{P_R}, \recvCtr_{P_R} = \recvTag.\ack.(\sendCtr, \recvCtr)$\\
        \ind $v_s = (S, \sendCtr_P, \recvCtr_P, m)$\\
        \ind $v_r = (R, \sendCtr_{P_R}, \recvCtr_{P_R}, m)$\\
        \ind $\cif v_s \not\in G.V_P$\\
        \ind\ind  Add $v_s$ to $G.V_P$\\
        \ind Add $v_r$ to $G.V_{P_R}$, add $(v_s, v_r)$ to $G.E$\\
        $\creturn G$
    }
    \caption{Pseudocode for our $N$-party transcript franking construction.}
    \label{fig:tf-gm-const}
\end{figure*}

\paragraph{Security analysis}
The security analysis of our group construction closely mirrors that of the
two-party construction. As with the two-party construction, our group
construction achieves perfect reportability because the $\Rcv$ procedure
performs the same commitment checks as the $\verifyReport$ procedure. Below, we
prove the transcript integrity of our scheme.

\begin{theorem}      \label{thm:gm-tr-int} Let $\tfScheme$ be the group
      transcript franking scheme given in \figref{fig:tf-gm-const}. Let $\advA$
      be an $N$-party group transcript integrity adversary against $\tfScheme$.
      Then we give an EUF-CMA adversary $\advB$ and V-Bind adversary $\advC$ such
      that
          \begin{align*}
                \advantageTrInt_{\tfScheme, N}(\advA) \leq \advantageEUFCMA_{\mathsf{MAC}}(\advB)
                 + \advantageVBind_{\mathsf{CS}}(\advC)\;.
          \end{align*}
          Adversaries $\advB$ and $\advC$ run in time that of $\advA$ plus a small
          overhead.
    \end{theorem}

\begin{proof}
    We proceed as with the proof of the two-party construction. The games
    $\game_0$, $\game_1$, and $\game_2$ are defined as before, with a similar
    security argument. Here, we highlight notable differences in the group case.
    The additional bookkeeping $\R'$ is initialized to $\setelems{}$, as before,
    at the start of the game. Game $\game_0$ adds $(S, P_S, P_R, c.c_f,
    G.(\sendCtr_P, \recvCtr_P))$ to $\R'$ before the $\creturn$ statement of
    $\SendTagOracle$. Similarly, it adds $(R, P_S, c.c_f, G.(\sendCtr_P,
    \recvCtr_P))$ to $\R'$ before the $\creturn$ statement of $\RecvTagOracle$.
    The failure events $F_1$ and $F_2$ are defined as in the proof of
    \thmref{thm:dm-tr-int}.

    We now demonstrate an adversary $\advB$ where $\advantageEUFCMA(\advB) =
      \Pr[F_1]$. The adversary $\advB$ perfectly simulates $\game_0$ to $\advA$,
      while routing $\mac$ and $\verify$ calls to its challenger oracles. If
      $F_1$ occurs, then we have that $\widetilde{G_i} \not\subseteq G$ for some
      $i \in \setelems{1, 2}$. For $r$ an element of $\reportInfo$, define $f(r)
      = \setelems{r.\sendTag.\ack, r.\recvTag.\ack}$. Observe that
      $\widetilde{G_i} \not\subseteq G$ implies there is some $r^* = (P_S^*,
      P_R^*, m^*, k_f^*, c_f^*, t_s^*, t_r^*) \in \reportInfo_i$, such that
      $f(r^*) \not\subseteq \R'$.
      
      This results from the fact that $\bigcup_{r' \in \reportInfo} f(r')
      \subseteq \R'$ and $G' \gets \verifyReport(\reportInfo)$ implies
      $\widetilde{G'} \subseteq G$. We have that $G.(\sendCtr_P, \recvCtr_P) =
      (0, 0)$ for $P \in [N]$ and the server counters $\servSt.(\sendCtr_P,
      \recvCtr_P) = (0, 0)$ for $P \in [N]$ at the start of the game. When
      $\SendTagOracle$ is called, we increment $G.\sendCtr_P$ and
      $\servSt.\sendCtr_P$. Similarly, when $\RecvTagOracle$ is invoked, we
      increment $G.\recvCtr_P$ and $\servSt.\recvCtr_P$. This means that $v \in
      \widetilde{G'}.V$ implies $v \in G.V$ and $e \in \widetilde{G'}.E$ implies
      $e \in G.E$. This allows us to produce a forgery as shown in the proof for
      the two-party case.

      Now, we construct adversary $\advC$ where
      $\advantageVBind_{\commScheme}(\advC) = \Pr[F_2]$. If $F_2$ occurs, we
      have that the adversary $\advA$ was able to trigger $G_1 \not\approx G_2$
      while $\widetilde{G_1}, \widetilde{G_2} \subseteq G$. It is
      straightforward to see that the group version of
      \lemmaref{lemma:consistency} holds, so there exists $v_1 \in G_1$ and $v_2
      \in G_2$ such that $v_1[0:2] = v_2[0:2]$ but $v_1.m \neq v_2.m$. We
      proceed as in the proof of \thmref{thm:dm-tr-int} to produce a
      binding violation.\qed
\end{proof}

\section{Discussion and Extensions}
\label{sec:discussion}

\paragraph{Composing causality preservation with transcript franking}
Our work is primarily concerned with how malicious parties can interfere
with the reporting process. Causality preservation as proposed in
\cite{causalityPres} aims to model how parties in a messaging channel can obtain
consistent causality graphs in the presence of a malicious service provider.
Indeed, transcript franking can be instantiated with a channel that achieves
strong causality preservation to reap these benefits.

\paragraph{Reports with redacted messages}
We remark that our construction allows clients to report transmission patterns
of messages, via causality sub-graphs, without having to disclose every message
within the causality sub-graph. Doing so simply requires omitting the opening
key $k_f$ for the messages that a client wishes to redact within a report. In
Chen and Fischlin's construction of a message franking channel with causality,
this would not be possible as the causality metadata is committed to alongside
the plaintext message \cite{causalityPres}.

\paragraph{Malicious clients refusing acknowledgement}
In our protocol, clients indicate reception of a well-formed ciphertext to the
server before a reception event tag is created. Malicious clients may refuse to
make this acknowledgement to the server, preventing the sender of that message
from being able to report it. We can mitigate this via notifying senders of
message delivery and recipient validation separately, thereby flagging malicious
behavior in-band. Upon detecting this behavior, a client may refuse to further
communicate with the misbehaving client, but we do not yet support
cryptographically reporting this misbehavior.

To explain in more detail, when $P$ sends a message to $\barP$, there are three
key events in the course of message transmission: (1) the reception of the
message sent from $P$ by the platform server, (2) the reception of the message by
$\barP$, and (3) platform reception of a valid message acknowledgement from $\barP$.

If $\barP$ is malicious, it could refuse to indicate the validity of the
message, omitting step (3) as described above. When (an honest) $P$ notices that
only events (1) and (2) occurred for a particular message, while $\barP$
continues to send and acknowledge subsequent messages, it knows that $\barP$ is
acting in an aberrant manner and will halt interactions with $\barP$ (tear down
the conversation and alert the user). So detection of deviation is built into
our protocol. But our protocol does not enable $P$ to cryptographically prove to
a moderator that $\barP$ misbehaved in the way described above. A messaging
system might trust client software to report such misbehavior, but
cryptographically secure reporting of this class of misbehavior is an open
question.

The above issue is analogous to an honest client receiving a message with a
malformed franking tag in the standard single-message franking setting. Here the
recipient should drop the message, but cannot cryptographically prove to the
moderator that the sender sent a malformed ciphertext.

One might wonder if the key-committing aspect of the encryption scheme can come
to the rescue here: if a ciphertext can be decrypted only under one key, then
the sender needs to simply reveal the key in order to show the ciphertext is
well-formed. In messaging protocols, these keys are often intended for one-time
use and revealing them does not compromise forward-secrecy or post compromise
security. The issue is that there is no straightforward way to prove that the
recipient should have been able to derive a particular key from the ratcheting
protocol. Hence, proving that a message is decryptable is not sufficient to show
that the recipient should have been able to decrypt it. Designing protocols that
allow for proving such statements is an interesting future direction.

\paragraph{Deployment considerations}
Our definitions consider single conversations, however our constructions can be
parallelized in a straightforward manner for multiple conversations. Indeed, the
same server MAC key can be used, and as \cite{fbSecretConv} suggests,
periodically rotated. Tags will additionally have to include conversation
specific identifiers in order to ensure that messages cannot be falsely reported
across conversations. This would amount to appending the identifier $\convId$ to
$\ack$ before tagging it in $\tagSend$ and $\tagRecv$. Instead of using
numerical indices to identify parties, one might use unique user identifiers.
This is especially important for groups as membership can change over time,
hence so can the mapping between party indexes and user identities within a
group. Presenting causality information in a user-friendly way to messaging
parties and content moderators is an open question, which was also raised in
\cite{causalityPres}. Concretely, both the MAC and the commitment scheme can be
instantiated with HMAC-SHA-256 \cite{hmac}. A drawback of our proposed
construction is that the server must maintain counters for each party in each
channel. At scale, keeping track of this state can be prohibitive. In the
remainder of this section, we describe a modification of our scheme that
mitigates this issue.

\paragraph{Outsourced-storage transcript franking}
We now propose mechanisms for allowing clients to store the counters while the
server verifies how it is updated. In our no-server-storage solution, clients
send these counters in the associated data of their messages. We present a
solution for two-party transcript franking with outsourced storage in the
remainder of this section. In \appref{app:outsrc-security}, we provide a
security analysis of this solution and outline its generalization to group
transcript franking. We present pseudocode for our outsourced construction in
\figref{fig:o-tf-dm-const}, highlighting the procedures that differ from the
server-storage version.

Formally, an outsourced two-party transcript franking
scheme is a tuple of algorithms $\otf = (\srvInit,\allowbreak \init,\allowbreak
\Snd,\allowbreak \Rcv,\allowbreak \tagSend,\allowbreak \tagRecv,\allowbreak
\judge,\allowbreak \reportReplay)$, defined over a server state space $\stSp_S$,
a client state space $\stSp_C$, a key space $\keySp$, a message space $\msgSp$,
a commitment space $\commSp$, a franking key space $\keySp_f$, an initialization
tag space $\tagSp_I$, a message-sent tag space $\tagSp_S$, and a reception tag
space $\tagSp_R$. We detail these algorithms below:

\begin{itemize}
      \item $\servSt, t^{(0)}_0, t^{(1)}_0 \getsr \srvInit()$ outputs an initial
      server state $\servSt \in \stSp_S$, along with initialization tags
      $t^{(0)}, t^{(1)} \in \tagSp_I$ for each party.
      \item $\st \getsr \clInit(P, k)$ is defined as in \secref{sec:two-party}.
      \item $\st', c \getsr \Snd(P, \st, m)$ is defined as in
      \secref{sec:two-party}.
      \item $\servSt', \sendTag \gets \tagSend(\servSt, P, c_f, t)$ is a server
            procedure that produces a tag $\sendTag \in \tagSp_S$ for a message
            sending event, where $P$ is the sending party, $c_f$ is the franking
            tag, and $t \in \tagSp_S \cup \tagSp_R \cup
            \tagSp_I$ is the last tag issued for $P$.
      \item $\servSt, \recvTag \gets \tagRecv(\servSt, P, c_f, t)$ is a server
            procedure that produces a tag $\recvTag \in \tagSp_R$ for a message
            reception event by receiving party $P$. As with $\tagSend$, $t \in
            \tagSp_S \cup \tagSp_R \cup \tagSp_I$ is the last tag issued for $P$.
      \item $\st', m, k_f, i \gets \Rcv(\party, \st, c)$ is defined as in
      \secref{sec:two-party}.
      \item $G \gets \verifyReport(\st_S, \reportInfo)$ is defined as in
      \secref{sec:two-party}.
      \item $P \gets \reportReplay(\servSt, t, t')$ is a server procedure that
      takes as input two tags $t, t' \in \tagSp_S \cup \tagSp_R \cup \tagSp_I$.
      It outputs a party $P \in \bits$ if it determines that the tags constitute
      a replay attempt by $P$, or $\bot$ if no replay attempt is detected.
  \end{itemize}

\begin{figure*}[t]
    \centering
    \fhpages{0.47}{
        \procedurev{$\srvInit()$}\\
        $\macKey \getsr \keySp$; $\sendCtr_0, \recvCtr_0, \sendCtr_1, \recvCtr_1
        \gets (0, 0, 0, 0)$\\
        For $P \in \bits$\\
        \ind $t_0^{(P)}.\ack = (\texttt{Init}, P, \bot, 0, 0)$\\
        \ind $t_0^{(P)}.\tagLabel = \mac(\macKey, t_0^{(P)}.\ack)$\\
        $\creturn \{\macKey, \sendCtr_0, \recvCtr_0, \sendCtr_1, \recvCtr_1\}$,
        $t_0^{(0)}, t_0^{(1)}$\\

        \procedurev{$\tagSend(\servSt, P, c_f, t)$}\\
        $\cif \Pi(t) \neq P \lor \verify(\macKey, t.\ack, t.\tagLabel) = 0$:\\\ind $\creturn
        \servSt, \bot$\\
        $(\sendCtr_P, \recvCtr_P) \gets t.\ack.(\sendCtr, \recvCtr)$\\
        $\sendCtr_P \gets \sendCtr_P + 1$, $\ack \gets (S, P, \barP, c_f,
            \sendCtr_P, \recvCtr_P)$\\
        $\sendTag \gets (\ack, \mac(\macKey, \ack))$\\
        $\creturn \st_S, \sendTag$
    }{
        \procedurev{$\tagRecv(\servSt, P, c_f, t)$}\\
        $\cif \Pi(t) \neq P \lor \verify(\macKey, t.\ack, t.\tagLabel) = 0$:\\\ind $\creturn \servSt,
        \bot$\\
        $(\sendCtr_P, \recvCtr_P) \gets t.\ack.(\sendCtr, \recvCtr)$\\
        $\recvCtr_P \gets \recvCtr_P + 1$, $\ack \gets (R, \barP, P, c_f,
            \sendCtr_P, \recvCtr_P)$\\
        $\recvTag \gets (\ack, \mac(\macKey, \ack))$\\
        $\creturn \st_S, \recvTag$\\

        \procedurev{$\reportReplay(\servSt, t, t')$}\\
        $\cif \Pi(t) \neq \Pi(t') \cthen \creturn \bot$\\
        $b \gets (\verify(\macKey,\allowbreak t.\ack,\allowbreak t.\tagLabel) =
        1)\land\\\ind(\verify(\macKey,\allowbreak t'.\ack,\allowbreak
        t'.\tagLabel) = 1)\land\\\ind((t.\ack.\sendCtr + t.\ack.\recvCtr) =\\\ind
        (t'.\ack.\sendCtr + t'.\ack.\recvCtr))$\\
        $\cif b = 1 \cthen \creturn \Pi(t)$\\
        $\celse \creturn \bot$
        }
        \caption{Pseudocode for our two-party transcript franking construction
         with outsourced storage. Let $\Pi(t)$ be the sending party if $t$ is a
        sending tag and the receiving party if $t$ is a reception tag. The
        routines $\init$, $\Snd$, $\Rcv$, and $\judge$ remain unchanged relative
        to the pseudocode given in \figref{fig:tf-dm-const}.}
    \label{fig:o-tf-dm-const}
\end{figure*}

\paragraph{Preventing fast-forwards}
When sending a message, a client increments its send counter and appends to the
causality data a server tag for the previous counter. This prevents the client
from incrementing the counter too far into the future, resulting in what we term
a \emph{fast-forward} attack. Doing so would require the client to forge a MAC,
since it would have to produce a valid server tag on a message the server had
not tagged before. Syntactically, this means we modify $\tagRecv$ and $\tagSend$
to take in an additional input $t$, the latest tag issued by the server to party
$P$. At the initialization of a conversation, the server provides each party $P$
with a special starting tag $t_0^{(P)}$, which are additional outputs from
$\srvInit$, where $t_0^{(P)}.\ack = (\texttt{Init}, P, \bot, 0, 0)$. We write
$\tagSend(\servSt, P, c_f, t)$ and $\tagRecv(\servSt, P, c_f, t)$. In our
modified construction, the server first checks that $t$ is a valid tag and
obtains the initial values of the counters as $(\sendCtr_P, \recvCtr_P) \gets
t.\ack.(\sendCtr, \recvCtr)$ instead of retrieving them from its own storage,
for both $\tagSend$ and $\tagRecv$. If the check on $t$ fails,
$\tagRecv$ and $\tagSend$ output $\bot$.

\paragraph{Preventing replays}
If a client attempts to send a message with a repeated past counter, an honest
recipient can report the repeated counters to the server as proof of sender
misbehavior. Such a report provides resistance against rollback attacks for
counters. For message reception tags, we follow the same exact approach as it
applies to receive counters. We add a new procedure to the construction: $P
\gets \reportReplay(\servSt, t, t')$, where $P$ is the party that attempted the
replay. If the provided tags do not constitute proof of a replay, then the
output is $\bot$. In our modified construction, the server returns $P$ if
$(\verify(\macKey,\allowbreak t.\ack,\allowbreak t.\tagLabel) = 1)$,
$(\verify(\macKey,\allowbreak t'.\ack,\allowbreak t'.\tagLabel) = 1)$,
$((t.\ack.\sendCtr + t.\ack.\recvCtr) = (t'.\ack.\sendCtr + t'.\ack.\recvCtr))$,
and $(t.\ack \neq t'.\ack)$, where $P$ is the sending party within a send
acknowledgement or the receiving party of a reception acknowledgement -- if
these are not the same between $t$ and $t'$, then we output $\bot$. To submit a
false replay report framing an honest party would require a client to forge a
MAC. We term the submission of a false replay a \emph{replay framing attack},
for which we provide a game-based definition in \appref{app:outsrc-security}.

\paragraph{Replay reportability}
In addition to ensuring that honest clients cannot be falsely framed for
attempting a replay, we must guarantee that actual replays by malicious clients
are reportable. For a party $P$ that has been issued a server tag $t \in
\tagSp_S \cup \tagSp_R \cup \tagSp_I$, that then generates $t'$ and $t''$ by
invoking $\tagRecv$ and/or $\tagSend$ with the same $t$ given as the previous
tag, $\reportReplay(\servSt,\allowbreak t',\allowbreak t'')$ must output $P$. 
This holds for the construction in \figref{fig:o-tf-dm-const} by the checks
performed in $\judgeReplay$.

\section{Related Work}
\label{sec:relwork}

\paragraph{Message franking}
Message franking has been studied in various settings. Symmetric message
franking provides a reporting solution when the platform houses the moderation
endpoint for receiving reports, and when sender and recipient identities are
known \cite{fbSecretConv,smfComm,invisSalamanders,msgFrChan}. Our own work is
situated within this setting. Asymmetric message franking (AMF) generalizes to
metadata private platforms and allows for third-party moderation \cite{amf,
hecate}. Recent work has also generalized AMFs to group messaging
\cite{groupAMF}. There are also proposed reporting mechanisms built from secret
sharing \cite{reportingSecShare}. All of these works consider message franking
at the single-message level.

\paragraph{Causality in cryptographic protocols}
Prior work has investigated incorporating causality in cryptographic channels
\cite{mar17,groupChannels}. Notably, recent work by Chen and Fischlin has
introduced stronger causality notions and shown how to combine them with message
franking protocols \cite{causalityPres}. However, as we have discussed, their
message franking formalism does not meet our goals for transcript franking, due
to its reliance on client-reported causality information and the inability for
reporters to disclose their own sent messages. See \appref{app:comparison} for
more details. The distributed systems literature has long considered the problem
of ordering events over communication networks via devising notions of logical
time \cite{lamportClock} and distributed snapshots \cite{distSnapshot}.

\paragraph{Cryptographic Abuse Mitigation}
In addition to enabling user-driven reporting, other cryptographic solutions
have been proposed for targeting abuse in encrypted messaging. For messaging
platforms that allow forwarding content, message trace-back is a cryptographic
primitive that allows a platform to determine the origin of harmful content
\cite{traceback,soureTrackingPeale,hecate}. Message franking concerns
user-driven content reporting. Recent work has also considered automated
reporting for messages that match a list of known harmful content
\cite{applePsi}. Such proposals have been met with strong criticism from privacy
advocates. Follow-on work has attempted to navigate the privacy vs. moderation
trade-off through added transparency and placing limitations on what content can
be traced \cite{pubVerifHash,traceIlleg}. Another line of work explores how
cryptography can be used to aid with user-blocking \cite{orca,snarkBlock}.

\section{Conclusion}
\label{sec:conclusion}
Existing treatments of message franking only consider how reporters can disclose
individual messages that they receive. This is insufficient for including
necessary context within reports. Our work provides definitions and
constructions for transcript franking, an extension of message franking
protocols that allows reporting sequences of messages with strong guarantees
over message ordering and contiguity. We generalize our results to multi-party
messaging and show how to securely outsource state to clients, allowing for more
practical deployment. How our techniques can be generalized to asymmetric
message franking, in order to be applicable to metadata-private and third-party
moderation settings, remains an interesting open problem.

\newpage\clearpage\newpage
\bibliographystyle{splncs04}
\bibliography{main.bib}
\newpage
\appendix

\section{Comparison with Causality Preservation}
\label{app:comparison}
Recent work by Chen and Fischlin considers the problem of assuring ordering
integrity for messages within cryptographic channels, as well as extending this
integrity to message franking \cite{causalityPres}. They introduce a security
notion called \emph{causality preservation}, which captures the ability for two
parties to recover a consistent causal dependency graph over the messages they
exchange, even in the presence of a malicious network provider. To achieve
causality preservation, clients self-report the order in which messages were
sent and received via additional metadata. As a result, clients obtain a shared
view of the partial ordering in which messages were sent and received. This
partial ordering is captured by a causality graph, which we describe in
\secref{sec:prelim}. Our work additionally considers the problem of
multi-message franking in the group setting while Chen and Fischlin focus on
two-party messaging.

\paragraph{Overview of MFC with causality preservation} The causality metadata
is incorporated into a message franking scheme, enabling reporting of this order
in addition to the contents of the messages. Such context is valuable as the
ordering and contiguity of messages sent within a conversation can heavily
influence a moderator's interpretation of the reported messages. To illustrate
what this metadata looks like, we briefly recall the causal message franking
channel $\mfcFBCaus$ presented in the Chen-Fischlin paper.

A sender attaches causal metadata consisting of a queue $Q$ and an index $i_R$ to
each sent message. The queue $Q$ contains the actions performed by the sender
that have not yet been acknowledged by the recipient. Messages can be uniquely
identified by their sending index and party. An index with a bar $\bar{i}$
indicates a received message with sending index $i$. Actions recorded in $Q$
simply consist of these indices. Once the recipient indicates the latest message
it has received, the sender removes all elements from $Q$ up to and including
the one associated with that latest message. To communicate this, the index
$i_R$, sent alongside $Q$, indicates the largest message index received by the
sender from the other party, and the value $i_S$ keeps track of the largest
$\bar{i}_R$ received from the other party. Intuitively, the message with sending
index $\bar{i}_R$, now confirmed to have been received by the other party,
contains all elements of $Q$ up to and including the sending action $\bar{i}_R$,
allowing those actions to be safely removed from $Q$. Despite the optimization
that $i_S$ and $i_R$ enable, $Q$ can grow arbitrarily large, and the overall
bandwidth can increase in a quadratic manner if messages are not acknowledged.
In single-message franking, a user reports only messages they have received. In
causality-preserving message franking, the same is true, except these messages
contain metadata about the order in which other messages have been sent and
received.

\paragraph{Reporting self-sent message contents} Recall that the message
franking channel formalism allows you to only report the messages and content
received from the other party. This poses an issue for transcript franking since
we may have to report messages sent by the reporter themselves. Currently, The
Chen-Fischlin construction does not have a solution to this problem. Their
reporting formalism only allows reporting messages received from the other
party. One workaround would be to require an interactive reporting process in
which each party reports the messages of the other party, filling out the
causality graph. Of course, this is less than desirable, especially in the case
where an abusive party refuses to participate. In our solution, we propose a way
in which clients can explicitly acknowledge reception of well-formed messages,
thereby allowing the senders of those messages to independently report them.

\paragraph{Misbehavior by malicious clients} Via $Q$, clients self-report orderings
of message sending and reception events. An issue here is that clients can
self-report arbitrary such orderings, even ones that do not align with how
messages were transmitted through the service provider.
\figref{fig:message_order} illustrates an example for which arbitrary
self-reported messages orderings can enable malicious behavior.

Beyond the ability to deviate from the actual interleaving of messages, a
malicious sender or reporter can lie about messages having been dropped or
delivered out of order. For these reasons, relying on client-reported orderings
is insufficient for reporting sequences of messages.

\paragraph{Attack on prior construction}
The augmented Facebook message franking channel construction presented in
\cite{causalityPres}, $\mfcFBCaus$, allows clients to report message orders that
do not align with the ground truth. We illustrate this via a simple attack that
mirrors the idea presented in \figref{fig:message_order}. In order to provide a
fair comparison, we first describe a natural lifting from the Chen-Fischlin
notion for message franking channels to our setting. To instantiate a
Chen-Fischlin-style message franking channel in our setting, we define a
$\tagRecv$ function that simply returns $\bot$ for the tag. The $\judge$
function is defined in the natural way, by running $\Extr$ repeatedly for all
messages involved in the report in order to construct the full graph.

The adversary $\advA$ chooses Party 0 as the malicious party (this is opposite
of the attack presented in \secref{sec:intro}, but it applies in the same
manner) and issues the following oracle calls. We specify the exact causal
metadata that Party 0 embeds within the messages it sends to achieve this goal
within each send call.
\begin{enumerate}
    \item[] \textbf{Sequence 1:}
    \begin{enumerate}[label=\arabic*.]
        \item $c_1 \getsr \SendTagOracle(0, m_1; Q=(), i_R = -1)$
        \item $\RecvTagOracle(1, c_1)$
        \item $c_2 \getsr \SendTagOracle(1, m_2; Q=(\bar{1}), i_R = 1)$
        \item $\RecvTagOracle(0, c_2)$
        \item $c_3 \getsr \SendTagOracle(0, m_3; Q=(1), i_R = -1)$
        \item $\RecvTagOracle(1, c_3)$
        \item $c_4 \getsr \SendTagOracle(0, m_4; Q=(1, 2, \bar{1}), i_R = 1)$
        \item $\RecvTagOracle(1, c_4)$
    \end{enumerate}
\end{enumerate}

The adversary $\advA$ embeds causality metadata that suggests Party 0
having observed the ordering $(m_1, m_3, m_2, m_4)$. Meanwhile, Party 1
observes the ordering $(m_1, m_2, m_3, m_4)$, which also aligns with what
Party 0 should have observed had it not deviated from the protocol. The
ordering of $\SendTagOracle$ and $\RecvTagOracle$ calls differs from the
ordering specified in the $Q$ sent along with $m_4$. Below, we show the
causality metadata Party 0 would have attached had it followed the protocol
honestly:

\begin{enumerate}
    \item[] \textbf{Sequence 2:}
    \begin{enumerate}[label=\arabic*.]
        \item $c_1 \getsr \SendTagOracle(0, m_1; Q=(), i_R = -1)$
        \item $\RecvTagOracle(1, c_1)$
        \item $c_2 \getsr \SendTagOracle(1, m_2; Q=(\bar{1}), i_R = 1)$
        \item $\RecvTagOracle(0, c_2)$
        \item $c_3 \getsr \SendTagOracle(0, m_3; Q=(\bar{1}), i_R = 1)$
        \item $\RecvTagOracle(1, c_3)$
        \item $c_4 \getsr \SendTagOracle(0, m_4; Q=(\bar{1}, 2), i_R = 1)$
        \item $\RecvTagOracle(1, c_4)$
    \end{enumerate}
\end{enumerate}

As a result, the graph recovered from $\Extr$ differs from the graph
maintained by the security game. The adversary $\advA$ can issue a report, from
either party, indicating the wrong causal ordering, winning with probability 1.

\paragraph{Impossibility result} The attack we just presented works because the
server has no way to tag reception events. Since the message franking channel
presented in \cite{causalityPres} does not enable such tagging, it is impossible
for any scheme within their model to satisfy our transcript integrity security
notion. Intuitively, any scheme that doesn't enable the server to certify when
messages are received requires the recipient party to self-report when reception
events occur. In particular, the attack we just presented generalizes to any
scheme that does not enable the server to tag reception events. Hence, we extend
the message franking model to allow the server to tag reception events, via that
$\tagRecv$ procedure.

\begin{theorem}
    Any message franking channel $\mfc$ that does not tag message reception
    events does not satisfy transcript integrity.
\end{theorem}
\begin{proof}
    We consider Sequence 1 and Sequence 2 as defined above and note that they
    provide a scenario in which the sending events occur in the same order, but
    the reception events occur in a different order. Given that client input
    cannot be trusted to report the ordering of reception events, the $\judge$
    routine needs to somehow recover the causality graph given just information
    about the order in which messages were sent, which it can record upon
    invocation of $\tagSend$. Since Sequences 1 and 2 have the same ordering of
    send operations ordering, but correspond to different causality graphs due
    to the difference in reception ordering, we see that it is impossible for
    $\judge$ to distinguish between these two scenarios in general. In the
    transcript integrity game, the adversary can randomly choose to execute
    Sequence 1 or Sequence 2 with probability $1/2$. For any fixed $\judge$
    routine, the probability of outputting the correct causality graph is at
    most $1/2$. Hence, no scheme that fails to consider reception events can
    achieve our notion of transcript integrity. \qed
\end{proof}

\section{Security for Outsourced-storage Transcript Franking}
\label{app:outsrc-security}
In this section, we elaborate on the security analysis of our outsourced
transcript franking scheme, which we introduced in \secref{sec:discussion}.
Correctness for outsourced transcript franking is captured by the game in
\figref{fig:o-correctness}. We adapt the security notions for
server-side-storage transcript franking to the outsourced setting and prove that
our outsourced scheme achieves security.

\paragraph{Transcript integrity} In \figref{fig:o-tr-integ}, we present our
transcript integrity definition for outsourced transcript franking. The goal of
the adversary is the same as in the non-outsourced transcript integrity game,
except now the adversary must submit valid server tags to generate new ones.
Moreover, the adversary cannot replay tags in a way that is undetected by
$\reportReplay$. 
The advantage of an adversary $\advA$ in the outsourced storage transcript
integrity game is
$$\advantageOTrInt_{\tfScheme}(\advA) = \Pr[\gameOTrInt_{\tfScheme}(\advA) =
1]\;.$$

\paragraph{Replay framing} In \figref{fig:o-frame}, we present a security game
for replay framing. Given two parties that honestly interact in the protocol,
the adversary attempts to generate a replay that is detected by $\reportReplay$,
in effect framing an honest party for attempting to replay a server tag. The
advantage of an adversary $\advA$ in the outsourced storage replay framing game
is
$$\advantageOFrame_{\tfScheme}(\advA) = \Pr[\gameOFrame_{\tfScheme}(\advA) = 1]\;.$$

\begin{figure*}[t]
    \centering
    \fhpages{0.4}{
        \procedurev{$\gameCorr_{\tfScheme}(\advA)$}\\
        $\serverKey \getsr \keySp$,
        $\st_\advA, \channelKey \getsr \advA()$,
        $\winFlag \gets 0$\\
        $\st_S, t_0^{(0)}, t_0^{(1)}\getsr \srvInit(\serverKey)$\\
        $\st_0 \getsr \clInit(0, \channelKey)$,
        $\st_1 \getsr \clInit(1, \channelKey)$\\
        $t_0, t_1 \gets t_0^{(0)}, t_0^{(1)}$\\
        $\R_t, \recvMessages \gets \setelems{}, \setelems{}$\\
        $\advA^\calO(\st_\advA, \channelKey)$\\
        $\creturn \winFlag$\\

        \procedurev{$\calO.\SendTagOracle(P, m)$}
        $(\st_P, c) \getsr \Snd(P, \st_P, m)$\\
        $\st_S, t_P \getsr \tagSend(\st_S, P, c.c_f, t_P)$\\
        $G \gets G + (\texttt{S}, P, m)$\\
        Add $(P, c, t_P)$ to $\R_t$\\
        $\creturn c, t_P$
    }{

        \procedurev{$\calO.\RecvTagOracle(P, c, \sendTag)$}\\
        Assert $(\barP, c.c, \sendTag) \in \R_t$\\
        $\st_P, m, k_f, i \gets \Rcv(P, \st_P, c)$\\
        $\cif m \neq \bot \cthen$\\
        \ind$\st_S, t_P \getsr \tagRecv(\st_S, P, c_f, t_P)$\\
        \ind $G \gets G + (\texttt{R}, P, c.i)$\\
        \ind Add $(P, m, k_f, c_f, t_s, t_r)$ to $\recvMessages$\\
        $\celse$\\
        \ind $\winFlag \gets 1$\\
        $\creturn m, k_f, \sendTag, t_P$\\

        \procedurev{$\calO.\ReportOracle(\reportInfo)$}\\
        Assert $|\reportInfo| > 0$
        $G' = \verifyReport(\servSt, \reportInfo)$\\
        $\cif \reportInfo \subseteq \R_r \land ((G' = \bot) \lor (G' \not\subseteq
        G))$:\\
        \ind $\winFlag \gets 1$
    }
    \caption{The security game for outsourced-storage transcript franking correctness.
    }
    \label{fig:o-correctness}
\end{figure*}

\begin{figure*}[t]
    \centering
    \fhpagesss{0.23}{0.35}{0.35}{
        \procedurev{$\gameOTrInt_{\tfScheme}(\advA)$}\\
        $\serverKey \getsr \keySp$; $\winFlag \gets 0$\\
        $\st_\advA, \channelKey \getsr \advA()$\\
        $\st_S, t_0^{(0)}, t_0^{(1)}\getsr\\\ind \srvInit(\serverKey)$\\
        $\st_0 \getsr \clInit(0, \channelKey)$\\
        $\st_1 \getsr \clInit(1, \channelKey)$\\
        $G, \R_r, \R \gets \varepsilon, \setelems{}, \setelems{}$\\
        $\advA^\calO(\st_\advA, \channelKey, t_0^{(0)}, t_0^{(1)})$\\
        $\creturn \winFlag$
    }{
        \procedurev{$\calO.\SendTagOracle(P, c, t)$}\\
        Assert for all $t_1, t_2 \in \R$,\\\ind $\reportReplay(\servSt, t_1, t_2) = 0$\\
        $\st_S, \sendTag \gets \tagSend(\st_S, P, c.c_f, t)$\\
        $\cif \sendTag = \bot \cthen \creturn \bot$\\
        $G \gets G + (\texttt{S}, P)$\\
        Add $(P, c, \sendTag)$ to $\R_t$, add $\sendTag$ to $\R$\\
        $\creturn \sendTag$
    }{
        \procedurev{$\calO.\RecvTagOracle(P, c, \sendTag, t)$}\\
        Assert $(\barP, c, \sendTag) \in \R_t$\\
        Assert for all $t_1, t_2 \in \R$,\\\ind $\reportReplay(\servSt, t_1, t_2) = 0$\\
        $\st_S, \recvTag \gets \tagRecv(\st_S, P, c.c_f, t)$\\
        $\cif \recvTag = \bot \cthen \creturn \bot$\\
        $G \gets G + (\texttt{R}, P, c.i)$\\
        Add $\recvTag$ to $\R$\\
        $\creturn \recvTag$\\

        \procedurev{$\calO.\ReportOracle(\reportInfo_1, \reportInfo_2)$}\\
        Assert $|\reportInfo_1| > 0$ and $|\reportInfo_2| > 0$\\
        $G_1 \gets \verifyReport(\servSt, \reportInfo_1)$\\
        $G_2 \gets \verifyReport(\servSt, \reportInfo_2)$\\
        $\cif G_1 \neq \bot \land G_2 \neq \bot \land\\
        \ind ((\widetilde{G_1} \not\subseteq G) \lor (\widetilde{G_2} \not\subseteq G)\\\ind \lor (G_1
        \not\approx G_2))$:\\
        \ind $\winFlag \gets 1$
    }
    \caption{The security game for outsourced-storage transcript integrity.}
    \label{fig:o-tr-integ}
\end{figure*}

\begin{figure*}[t]
    \centering
    \fhpagesss{0.23}{0.35}{0.35}{
        \procedurev{$\gameOFrame_{\tfScheme}(\advA)$}\\
        $\serverKey \getsr \keySp$; $\winFlag \gets 0$\\
        $\st_\advA, \channelKey \getsr \advA()$\\
        $\st_S, t_0^{(0)}, t_0^{(1)}\getsr\\\ind \srvInit(\serverKey)$\\
        $\st_0 \getsr \clInit(0, \channelKey)$\\
        $\st_1 \getsr \clInit(1, \channelKey)$\\
        $t_0, t_1 \gets t_0^{(0)}, t_0^{(1)}$\\
        $\R_r, \R, \tagggedMessages \gets \\\ind\setelems{}, \setelems{},
        \setelems{}$\\
        $\advA^\calO(\st_\advA, \channelKey, t_0^{(0)}, t_0^{(1)})$\\
        $\creturn \winFlag$
    }{
        \procedurev{$\calO.\RecvTagOracle(P, c, \sendTag)$}\\
        Assert $(\barP, c, \sendTag) \in \R_t$\\
        Assert $c \not\in \R$\\
        $\st_S, t_P \gets \tagRecv(\st_S, P, c.c_f, t_P)$\\
        $\cif t_P = \bot \cthen \creturn \bot$\\
        Add $c$ to $\R$\\
        $\creturn t_P$\\

    }{
        \procedurev{$\calO.\SendTagOracle(P, c_f)$}\\
        $\st_S, t_P \gets \tagSend(\st_S, P, c.c_f, t_P)$\\
        $\cif t_P = \bot \cthen \creturn \bot$\\
        Add $(P, c, t_P)$ to $\R_t$\
        $\creturn t_P$\\

        \procedurev{$\calO.\ReportReplayOracle(t, t')$}\\
        $P \gets \reportReplay(\servSt, t, t')$\\
        $\cif P \neq \bot$:\\
        \ind $\winFlag \gets 1$
    }

    \caption{The security game for outsourced storage replay framing.
    }
    \label{fig:o-frame}
\end{figure*}

\paragraph{Security proofs} We now provide proofs for the transcript integrity
and reply framing security of our outsourced storage transcript franking scheme.
The following theorems establish the security of our outsourced construction.

\begin{theorem}      \label{thm:o-dm-tr-int}
      Let $\tfScheme$ be our outsourced-storage transcript franking scheme in \figref{fig:o-tf-dm-const}.
          Let $\advA$ be a transcript integrity adversary against $\tfScheme$. Then we give
          EUF-CMA adversaries $\advB$ and $\advC$, and a V-Bind adversary
          $\advD$, 
                such that
          \begin{align*}
                \advantageOTrInt_{\tfScheme}(\advA) \leq \advantageEUFCMA_{\mathsf{MAC}}(\advB)
                 + \advantageEUFCMA_{\mathsf{MAC}}(\advC) + \advantageVBind_{\mathsf{CS}}(\advD)\;.
          \end{align*}
          Adversaries $\advB$, $\advC$, and $\advD$ run in time that of $\advA$
          plus a small overhead.
    \end{theorem}
\begin{proof}
      We proceed via a sequence of game hops. Define $\game_0$ to be the same as
      $\gameOTrInt_{\tfScheme}$ with the additional bookkeeping as defined in
      the proof of \thmref{thm:dm-tr-int}. Let $\game_1$ be the same, except we
      abort if at any point $G.(\sendCtr_P, \recvCtr_P) \neq
      \servSt.(\sendCtr_P, \recvCtr_P)$. Let $F_1$ denote this event. We have
      $|\Pr[\game_0(\advA) \Rightarrow 1] - \Pr[\game_1(\advA) \Rightarrow 1]|
      \leq \Pr[F_1]$, and we will construct $\advB$ such that
      $\advantageEUFCMA_{\macScheme}(\advB) = \Pr[F_1]$. We have that $\advB$
      simulates $\game_0$ to $\advA$ while routing $\mac$ and $\verify$ calls to
      its challenger oracles. Observe that the only way for $F_1$ to occur is
      for $\advA$ to produce a $t^* \not\in \R$, which means that $t^*$ was
      never queried to the $\mac$ oracle, hence $\advB$ outputs it as a forgery.
      
      The rest of the proof proceeds similarly to the proof of
      \thmref{thm:dm-tr-int}.\qed
\end{proof}

\begin{theorem}      \label{thm:o-fr}       Let $\tfScheme$ be our
      outsourced-storage transcript franking scheme in
      \figref{fig:o-tf-dm-const}. Let $\advA$ be a replay framing adversary
      against $\tfScheme$. Then we give an EUF-CMA adversary $\advB$  such
      that
          \begin{align*}
                \advantageOFrame_{\tfScheme}(\advA) \leq \advantageEUFCMA_{\mathsf{MAC}}(\advB)\;.
          \end{align*}
          Adversary $\advB$ runs in time that of $\advA$ plus a small overhead.
    \end{theorem}
\begin{proof}
      Our adversary $\advB$ perfectly simulates $\gameOFrame_{\tfScheme}$ to
      $\advA$ while routing calls to $\mac$ and $\verify$ to its own challenger
      oracles. For any two tags $t_1, t_2$ output by the $\mac$ oracle, we must
      have that $\reportReplay(\servSt, t_1, t_2) = \bot$, since these tags are
      the output of the honest $\tagRecv$ and $\tagSend$ procedures. Therefore,
      if $\advA$ wins, it must have produced a pair $(t, t')$ where at least one
      of these tags was not output by the $\mac$ challenger oracle. Let $t^*$ be
      that tag. The adversary $\advB$ outputs $t^*$ and wins with the same
      probability that adversary $\advA$ wins. \qed
\end{proof}

\paragraph{Group outsourced transcript franking} The construction and analysis
we provide for outsourced two-party transcript franking generalizes naturally to
the $N$-party group messaging setting. We sketch the necessary modifications
here. As with the two-party outsourced setting, we have that $\srvInit$, in
addition to $\servSt$, outputs a list of initial tags $t_0^{(0)}, \ldots,
t_0^{(N-1)}$. The server-side tagging procedures $\tagSend$ and $\tagRecv$
accept an additional argument $t$ for the previous tag issued to a party. We
include an additional procedure $\judgeReplay(\servSt, t, t')$, which outputs a
party $P$ if $(t, t')$ constitute a replay attack by $P$, or $\bot$ otherwise.
We outline our outsourced group transcript franking construction in
\figref{fig:o-tf-gm-const}.

\begin{figure*}[t]
    \centering
    \fhpages{0.47}{
        
    \procedurev{$\srvInit(N)$}\\
        $\macKey \getsr \keySp$\\
        For $P \in [N]$\\
        \ind $t_0^{(P)}.\ack = (\texttt{Init}, P, \bot, 0, 0)$\\
        \ind $t_0^{(P)}.\tagLabel = \mac(\macKey, t_0^{(P)}.\ack)$\\
        \ind $\sendCtr_i, \recvCtr_i \gets 0, 0$\\
        $\creturn \{\macKey\} \cup \setelems{\sendCtr_i, \recvCtr_i}_{i\in[N]}$,
        $\setelems{t_0^{(i)}}_{i\in[N]}$\\

        \procedurev{$\tagSend(\servSt, P, c_f, t)$}\\
        $\cif \Pi(t) \neq P \lor \verify(\macKey, t.\ack, t.\tagLabel) =
        0$:\\\ind $\creturn \servSt, \bot$\\
        $(\sendCtr_P, \recvCtr_P) \gets t.\ack.(\sendCtr, \recvCtr)$\\
        $\sendCtr_P \gets \sendCtr_P + 1$, $\ack \gets (S, P, c_f,
            \sendCtr_P, \recvCtr_P)$\\
        $\sendTag \gets (\ack, \mac(\macKey, \ack))$\\
        $\creturn \st_S, \sendTag$
    }{
        \procedurev{$\tagRecv(\servSt, P_R, P_S, c_f, t)$}\\
        $\cif \Pi(t) \neq P_R \lor \verify(\macKey, t.\ack, t.\tagLabel) =
        0$:\\\ind $\creturn \servSt, \bot$\\
        $(\sendCtr_{P_R}, \recvCtr_{P_R}) \gets t.\ack.(\sendCtr, \recvCtr)$\\
        $\recvCtr_{P_R} \gets \recvCtr_{P_R} + 1$\\
        $\ack \gets (R, P_S, P_R, c_f,
            \sendCtr_{P_R}, \recvCtr_{P_R})$\\
        $\recvTag \gets (\ack, \mac(\macKey, \ack))$\\
        $\creturn \st_S, \recvTag$\\

        \procedurev{$\reportReplay(\servSt, t, t')$}\\
        $\cif \Pi(t) \neq \Pi(t') \cthen \creturn \bot$\\
        $b \gets (\verify(\macKey,\allowbreak t.\ack,\allowbreak t.\tagLabel) =
        1)\land\\\ind(\verify(\macKey,\allowbreak t'.\ack,\allowbreak
        t'.\tagLabel) = 1)\land\\\ind((t.\ack.\sendCtr + t.\ack.\recvCtr) =\\\ind
        (t'.\ack.\sendCtr + t'.\ack.\recvCtr))$\\
        $\cif b = 1 \cthen \creturn \Pi(t)$\\
        $\celse \creturn \bot$
        }
        \caption{Pseudocode for our $N$-party transcript franking construction
         with outsourced storage. Let $\Pi(t)$ be the sending party if $t$ is a
         sending tag and the receiving party if $t$ is a reception tag.}
    \label{fig:o-tf-gm-const}
\end{figure*}

\end{document}